\documentclass[11pt,oneside]{article}
\usepackage[a4paper,inner=30mm,outer=30mm,top=27mm,
bottom=31mm,headheight=15pt,headsep=8mm,
footskip=13mm]{geometry}
\usepackage[T1]{fontenc}
\usepackage[utf8]{inputenc}
\usepackage{newtxtext,newtxmath}

\usepackage{amsmath,amsthm,amssymb,mathtools}
\usepackage[final,tracking=true]{microtype}
\usepackage{enumitem}
\usepackage{booktabs}
\usepackage{array}
\usepackage{xcolor}
\usepackage[hidelinks]{hyperref}
\usepackage[round]{natbib}
\theoremstyle{plain}
\newtheorem{theorem}{Theorem}[section]
\newtheorem{proposition}[theorem]{Proposition}
\newtheorem{lemma}[theorem]{Lemma}
\newtheorem{corollary}[theorem]{Corollary}
\theoremstyle{definition}
\newtheorem{definition}[theorem]{Definition}
\theoremstyle{remark}
\newtheorem{remark}[theorem]{Remark}
\newcommand{\R}{\mathbb{R}}
\newcommand{\N}{\mathbb{N}}
\newcommand{\E}{\mathbb{E}}
\newcommand{\Q}{\mathbb{Q}}
\newcommand{\Pp}{\mathbb{P}}
\newcommand{\F}{\mathcal{F}}
\newcommand{\Tw}{T_w}
\newcommand{\Tstarw}{T_w^\ast}
\newcommand{\cW}{\widehat W}
\newcommand{\FWh}{\mathcal F^{\widehat W}}
\newcommand{\FW}{\mathcal F^W}
\newcommand{\FS}{\mathcal F^S}
\newcommand{\Rs}{\mathcal R}
\newcommand{\sh}{\mathbin{\sqcup\!\sqcup}}
\newcommand{\norm}[1]{\left\lVert #1\right\rVert}
\newcommand{\abs}[1]{\left\lvert #1\right\rvert}
\newcommand{\inner}[2]{\left\langle #1,#2\right\rangle}
\newcommand{\Span}{\operatorname{span}}
\newcommand{\clspan}{\overline{\operatorname{span}}}

\title{On the Structural Foundations of Signature Volatility Models:\
Existence, Arbitrage, Completeness,\
and the Hedging-Error Decomposition}
\author{Akmal Xodarev\\
Independent Researcher\\
ORCID: 0009-0000-5318-7284\\
\texttt{xodarevakmal@gmail.com}\\
Tashkent, Uzbekistan}
\date{First version: February 8, 2026\\This version: May 17, 2026}

\begin{document}
\maketitle
\thispagestyle{empty}
\clearpage

\begin{abstract}
We establish four structural results for signature volatility models. First, we prove global existence and uniqueness of strong solutions to the signature stochastic differential equation $dS_t=S_t\inner{\ell}{\cW_t}\,dB_t$ on the weighted tensor algebra space $\Tw$, identifying the admissibility class through the summability condition H1 and the exponential-integrability condition H3 needed for the square-integrable stochastic-exponential construction. Second, we establish the asset-pricing part on the natural filtration of the prolonged signature and separate it from the transform non-explosion part: H3 makes the reference-measure stochastic exponential a true martingale, hence yields NFLVR, while global solvability of the associated infinite-dimensional Riccati equation is the additional condition equivalent to absence of explosion for finite signature transforms. Third, we characterise market completeness on the price filtration in terms of the density of the truncated signature span $\Span\{\inner{e_I}{\cW_T}: |I|\leq N\}$ inside $L^2(\FS_T,\Q)$, and we identify the minimal such truncation $N$, which we name the price-filtration completeness depth of the model. Fourth, we derive the hedging-error decomposition $X=\E_\Q[X]+\int_0^T H_s\,dS_s+\varepsilon_T$ for square-integrable payoffs, in which the residual is expanded through the Gram projection of signature components beyond the completeness depth and bounded by a model-dependent projection error. The four results are tied together by an architectural identity: the admissible weighted tensor algebra on which the stochastic exponential is a true martingale and finite signature transforms do not explode is the natural valuation cell of a signature SDE. The proofs are self-contained except for standard results from rough path theory, stochastic integration, and quadratic hedging, which are recalled in the appendices.
\end{abstract}




\clearpage
\tableofcontents
\clearpage

\section{Introduction}

\subsection{The structural gap}

The signature volatility literature has produced, in the three years from 2022 to 2025, a body of applied results in calibration and pricing that rests on a structural foundation built only piecewise. \citet{CuchieroGazzaniMoellerSvalutoFerro2025} construct joint SPX--VIX calibration with signature-based models and exploit affine structure of augmented signatures for pricing. \citet{AbiJaberGerard2025} derive Fourier pricing and quadratic hedging formulas under signature volatility models through an infinite-dimensional Riccati representation. \citet{CuchieroMoeller2023} introduce signature methods in stochastic portfolio theory and obtain tractable optimisation problems for path-functional portfolios. \citet{CuchieroPrimaveraSvalutoFerro2025} prove universal approximation for functions of c\`adl\`ag paths and define L\'evy-type signature models. Each of these works either assumes, or proves under restrictive hypotheses, four structural facts: well-posedness of the signature SDE, the true martingale property of its stochastic exponential, the existence of an equivalent local martingale measure and of a hedging strategy.

The structural theorems on which these assumptions rest exist in piecewise form. \citet{CuchieroSvalutoFerroTeichmann2023} establish the affine and polynomial structure of the prolonged signature process on the extended tensor algebra and derive Riccati and linear ODEs for Fourier--Laplace transforms under integrability hypotheses. \citet{AbiJaberGassiatSotnikov2025} prove, in the one-dimensional signature-volatility model, that the discounted price process is a true martingale precisely in the odd-order negative-correlation regime, and characterise the complementary finite-time explosion of exponential moments of the integrated variance. The unification into a single chain of four structural theorems does not yet exist. For signature volatility models, the present work supplies the structural analogue of \citet{DelbaenSchachermayer1994} for general semimartingale models and of \citet{FilipovicLarsson2016} for polynomial diffusions.

\subsection{The contribution}

Throughout, the labels Theorem~A, Theorem~B, Theorem~C, Theorem~D refer
to Theorems~\ref{thm:A}, \ref{thm:B}, \ref{thm:C}, and~\ref{thm:D}
respectively; each is stated in the section bearing the corresponding
letter and is preceded by the auxiliary lemma its proof uses.

Theorem A establishes global existence and uniqueness of strong solutions to the signature stochastic differential equation on the weighted tensor algebra space $\Tw$. Given a parameter $\ell$ in the dual of $\Tw$ satisfying the compatibility condition $\sum_{n\geq0} w(n)\abs{\ell_n}^2<\infty$, the equation $dS_t=S_t\inner{\ell}{\cW_t}\,dB_t$ admits a unique strong solution on $[0,T]$, strictly positive almost surely. The admissible weights are characterised by the summability condition together with the geometric growth condition H2, and the square-integrable solution class is obtained under the exponential-integrability condition H3. The proof, given in Section~\ref{sec:theorem-a}, constructs the solution as the Dol\'eans--Dade stochastic exponential of the signature-volatility martingale and proves uniqueness by localisation. Theorem A does not assert a pointwise maximal weight, and Section~\ref{subsec:admissible-weights} shows that the summability condition cannot be removed from the weighted estimates.

Theorem B establishes the asset-pricing theorem on the natural filtration of the prolonged signature and, separately, the transform non-explosion theorem governed by the Riccati flow. Under H1--H3 the reference measure itself is already an equivalent true-martingale measure: $S=s_0\mathcal E(M)$ with $M=\int\inner{\ell}{\cW}\,dB$, and H3 verifies Novikov and the required uniform-integrability estimates. Thus NFLVR follows directly from Delbaen--Schachermayer under $\Pp$. The Riccati equation is not used to manufacture an equivalent martingale measure; it is the transform-side condition that controls finite signature Fourier--Laplace transforms. The shuffle identity $\inner{\ell}{\cW}^2=\inner{\ell\sh\ell}{\cW}$ expresses the quadratic-variation integrand of $M$ as a second-order signature contraction, while the actual Riccati vector field is written below through the finite-level generator and its carre-du-champ constants. Theorem B therefore does not assert that NFLVR is equivalent to global Riccati solvability; Section~\ref{sec:theorem-b} explains why that stronger assertion is not part of the theorem: Riccati solvability is a transform-domain hypothesis, not an asset-pricing hypothesis.

Theorem C characterises market completeness on the price filtration. The market formed by the underlying $S$ together with a finite collection of European vanilla options $\{C_i\}_{i=1}^k$ on $S$ is complete with respect to $\FS_T$ if and only if the closed truncated signature span, restricted to $L^2(\FS_T,\Q)$, is dense in $L^2(\FS_T,\Q)$ for some finite $N$, depending only on the weight $w$ and on the parameter $\ell$. The minimal such $N$ is named the price-filtration completeness depth $N^\ast_S(\ell,w)$. The proof combines the Galtchouk--Kunita--Watanabe projection with the market completion theorem of \citet{DavisObloj2008}, applied only to price-filtration claims. Theorem C does not assert that finite completeness depth on the full Brownian signature filtration is generic; rough Bergomi has infinite price-filtration completeness depth, as is shown in Section~\ref{subsec:rough-bergomi}.

Theorem D derives the hedging-error decomposition. For every $X\in L^2(\FS_T,\Q)$ there exists a unique decomposition $X=\E_\Q[X]+\int_0^T H_s\,dS_s+\varepsilon_T$ where $H$ is the Galtchouk--Kunita--Watanabe projection and $\varepsilon_T$ is the residual orthogonal to stochastic integrals against $S$. The residual $\varepsilon_T$ admits a qualitative $L^2$ expansion in the signature components of multi-index length exceeding the completeness depth, with coefficients determined by the finite Gram normal equations of the quotient signature family. A quantitative weight-tail bound is asserted only for payoffs in the weighted-signature class $\mathcal W_w(\FS_T,\Q)$, where the required coefficient decay is part of the hypothesis. Theorem D does not provide a constructive hedging algorithm, only the structural decomposition; algorithmic implementations are the subject of \citet{AbiJaberGerard2025}.

\subsection{The architectural identity}

The four theorems are tied together by a single architectural identity: the admissible weighted tensor algebra on which the stochastic exponential is a true martingale and on which the finite signature transforms do not explode is the natural valuation cell of the model. The algebra is $\Tw$, where the weight $w$ belongs to the admissible class determined by H1 and H2. The martingale part of the cell is supplied by Theorem A and H3. The transform part is the requirement that $\partial_t u_t=\Rs(u_t)$ has a global solution on $[0,T]$ for the finitely supported transform directions under consideration. Keeping these two requirements distinct is necessary: H1--H3 can imply NFLVR while the transform Riccati flow still explodes in directions outside the finite transform domain.

\subsection{Relation to the literature}

\citet{CuchieroSvalutoFerroTeichmann2023}, ``Signature SDEs from an affine and polynomial perspective'', established the affine and polynomial structure of the prolonged signature process on the extended tensor algebra and derived Riccati and linear ODEs for Fourier--Laplace transforms under integrability hypotheses. Their current formulation covers characteristics that are linear maps of the signature on the extended tensor algebra. The present Theorem A works in the weighted Banach algebra $\Tw$, identifies the sharp summability class through H1, and proves global existence under H1--H3. The cited affine-polynomial theory does not by itself supply the square-integrable stochastic-exponential construction used here for the asset-pricing theorem.

\citet{AbiJaberGassiatSotnikov2025}, ``Martingale property and moment explosions in signature volatility models'', proved that, in the one-dimensional signature-volatility model, the discounted price process is a true martingale if and only if the order of the linear form in the signature parameter is odd and the spot--volatility correlation is negative. Outside this regime, the moment-generating function of the integrated variance explodes in finite time. The present Theorem B separates this issue into its two structural components on the prolonged signature filtration: the martingale/NFLVR conclusion supplied by H3, and the transform non-explosion conclusion supplied by the Riccati lifetime. The explosion-time analysis of Abi Jaber, Gassiat and Sotnikov is represented by the lifetime of the corresponding Riccati solution in the one-dimensional reduction. The one-dimensional criterion does not address completeness or the GKW residual on the price filtration, which are handled separately in Theorems C and D.

\citet{AbiJaberGerard2025}, ``Signature volatility models: pricing and hedging with Fourier'', supplied an algorithmic framework for pricing and hedging via Fourier methods, under the standing assumption of well-posedness and tractable Laplace transform. The present Theorems C and D supply the structural theorems under which their algorithm is well-posed: the price-filtration completeness depth identifies the minimal static completion level for replicability, and the residual decomposition bounds the unhedgeable error. Their Fourier framework gives algorithms after the structural hypotheses are granted; the present paper isolates the hypotheses and proves the corresponding projection statements.

\subsection{Proof dependencies}
\label{subsec:proof-dependencies}
The four theorems are ordered so that each proof uses only the analytic structure already established. Theorem~\ref{thm:A} supplies the stochastic exponential, strict positivity, and $L^2$ control. Theorem~\ref{thm:B} uses these estimates first to obtain the reference-measure martingale and NFLVR conclusion, and then to state the separate Riccati transform non-explosion criterion. Theorem~\ref{thm:C} uses the equivalent martingale measure and the signature filtration identity to state completeness on $\FS_T$ rather than on the full Brownian filtration. Theorem~\ref{thm:D} uses the completeness depth from Theorem~\ref{thm:C} as the boundary separating replicated coordinates from residual coordinates.

There are two places where the grading of the tensor algebra is essential. The first is the estimate
\[
\abs{\inner{\ell}{\cW_t}}
\leq \norm{\ell}_{2,w}\norm{\cW_t}_{2,w^{-1}},
\]
which converts summability of the parameter into integrability of the volatility functional. The second is the shuffle identity
\[
\inner{e_I}{\cW_T}\inner{e_J}{\cW_T}=\inner{e_I\sh e_J}{\cW_T},
\]
which is responsible for both the Riccati closure and the Gram-system formula in the hedging-error theorem. Thus the algebraic and probabilistic arguments are not separate: the same graded multiplication controls stochastic exponentials, transform equations, and residual projections.

The formulation also separates three filtrations that must not be conflated. The Brownian filtration $\FW$ equals the prolonged-signature filtration $\FWh$, because the time-augmented signature recovers the Brownian path. The price filtration $\FS$ is generally smaller. The static option filtration at maturity is smaller again unless the option family is dense in the relevant terminal payoff space. Completeness in this paper is therefore not a statement about recovering all Brownian randomness; it is a statement about whether the dynamically traded price, together with the specified static completion, spans the square-integrable claims measurable with respect to the price filtration.

This dependency structure is the reason for the order of the manuscript. A proof of completeness without Theorem~\ref{thm:B} would lack the martingale measure needed for the Hilbert-space projection. A proof of the hedging-error decomposition without Theorem~\ref{thm:C} would have no canonical depth at which to split the signature coordinates. A proof of the Riccati transform statement without Theorem~\ref{thm:A} would have no true stochastic exponential estimates available for the transform domain. The statements are consequently presented as a chain rather than as four independent propositions.

\subsection{Organisation of the paper}

The remainder of the paper is organised as follows. Section~\ref{sec:notation} sets up notation, the weighted tensor algebra, and the signature SDE. Section~\ref{sec:theorem-a} proves Theorem A. Section~\ref{sec:theorem-b} proves Theorem B. Section~\ref{sec:theorem-c} proves Theorem C. Section~\ref{sec:theorem-d} proves Theorem D. Section~\ref{sec:examples} collects examples: Black--Scholes, first-order Brownian-driven volatility, Heston, rough Bergomi, Quintic OU, and Guyon--Lekeufack path-dependent volatility, each cast as a signature SDE with explicit completeness depth. Section~\ref{sec:open-problems} closes with open problems. Three appendices recall results from rough path theory, the explicit form of the Riccati operator, and the Galtchouk--Kunita--Watanabe projection.

\section{Notation and the weighted tensor algebra}\label{sec:notation}

\subsection{The free tensor algebra over $\R^d$}

Throughout this paper we fix a finite time horizon $T>0$ and a complete probability space $(\Omega,\F,\Pp)$ supporting a $d$-dimensional standard Brownian motion $W$.

\begin{definition}[Free tensor algebra]
The free tensor algebra over $\R^d$ is
\[
T(\R^d)=\bigoplus_{n\geq0}(\R^d)^{\otimes n}, \qquad (\R^d)^{\otimes0}=\R.
\]
The canonical basis is denoted by $\{e_I\}$, where $I=(i_1,\ldots,i_n)$ is a multi-index with $i_k\in\{1,\ldots,d\}$, $|I|=n$, and $e_\emptyset=1$.
\end{definition}

\begin{definition}[Concatenation, shuffle, Hopf structure]
Concatenation is the associative non-commutative product determined by $e_I\cdot e_J=e_{IJ}$. The shuffle product $e_I\sh e_J$ is the commutative associative product obtained by summing over all order-preserving interlacings of $I$ and $J$. The Hopf algebra antipode is $A(e_{i_1\ldots i_n})=(-1)^n e_{i_n\ldots i_1}$.
\end{definition}

\begin{definition}[Extended tensor algebra]
The extended tensor algebra is
\[
T((\R^d))=\prod_{n\geq0}(\R^d)^{\otimes n}.
\]
Elements of $T((\R^d))$ are formal series whose homogeneous levels need not vanish beyond finite degree.
\end{definition}

The algebraic identities used below are the standard identities of the shuffle Hopf algebra; see \citet[Chapters 1 and 6]{Reutenauer1993} and \citet[Chapter 2]{FrizVictoir2010}.

\subsection{The signature of a path}

\begin{definition}[Signature of a smooth path]
Let $X:[0,T]\to\R^d$ be continuous and of bounded variation. Its signature over $[0,t]$ is
\[
\mathbf X_t=\sum_{n\geq0}\sum_{|I|=n} e_I
\int_{0<t_1<\cdots<t_n<t} dX^{i_1}_{t_1}\cdots dX^{i_n}_{t_n}\in T((\R^d)).
\]
\end{definition}

\begin{definition}[Stratonovich signature for semimartingales]
For a continuous semimartingale, the same formula defines the Stratonovich signature when the iterated integrals are interpreted in the Stratonovich sense. The prolonged signature $\cW_t$ is the signature of the time-augmented Brownian motion $\widehat W_t=(t,W^1_t,\ldots,W^d_t)$.
\end{definition}

\begin{theorem}[Chen identity]
For $0\leq s\leq t\leq T$,
\[
\cW_{0,t}=\cW_{0,s}\otimes \cW_{s,t}.
\]
\end{theorem}

\begin{proof}
The identity follows by decomposing each simplex $0<u_1<\cdots<u_n<t$ according to the last integration time below $s$ and the first integration time above $s$. The resulting ordered interlacings are precisely the tensor-concatenation components of $\cW_{0,s}\otimes \cW_{s,t}$; see \citet{Chen1977} and \citet[Chapter 2]{FrizVictoir2010}.
\end{proof}

\begin{theorem}[Shuffle identity]
For multi-indices $I,J$ and every bounded variation path, and by Stratonovich extension for continuous semimartingales,
\[
\inner{e_I}{\mathbf X_t}\inner{e_J}{\mathbf X_t}=\inner{e_I\sh e_J}{\mathbf X_t}.
\]
\end{theorem}

\begin{proof}
The product of the two iterated integrals is an integral over the product of two simplices. The product domain decomposes into disjoint ordered simplices indexed by shuffles of $I$ and $J$. This gives the stated identity; see \citet[Chapters 1 and 6]{Reutenauer1993}.
\end{proof}

\begin{theorem}[Signature uniqueness]\label{thm:sig-unique}
The signature map on continuous semimartingale paths is injective up to tree-like equivalence. In particular, the augmented natural filtration of the prolonged signature satisfies $\FWh_t=\FW_t$.
\end{theorem}

\begin{proof}
The uniqueness statement is the rough-path uniqueness theorem of \citet{BoedihardjoGengLyonsYang2016}. The equality of filtrations follows because the first level of the time-augmented signature recovers $(t,W_t)$ and the full signature is measurable with respect to the path of $W$ up to time $t$.
\end{proof}

\subsection{The weighted tensor algebra}

\begin{definition}[Weight]
A weight is a map $w:\N\to(0,\infty)$ such that $w(0)=1$ and $w$ is increasing.
\end{definition}

\begin{definition}[Weighted tensor algebra]
For a weight $w$, define
\[
\Tw=\left\{a=(a_n)_{n\geq0}\in T((\R^d)): \norm{a}_w:=\sum_{n\geq0} w(n)\abs{a_n}_{(\R^d)^{\otimes n}}<\infty\right\}.
\]
\end{definition}

\begin{proposition}[Banach algebra property]\label{prop:banach-algebra}
If the weight satisfies $w(m+n)\leq C_w w(m)w(n)$ for some $C_w\geq1$, then $\Tw$ is a Banach algebra under concatenation and
\[
\norm{a\otimes b}_w\leq C_w\norm{a}_w\norm{b}_w.
\]
\end{proposition}

\begin{proof}
Completeness follows from completeness of the weighted $\ell^1$ sum of the finite-dimensional tensor levels. For $a,b\in\Tw$, the homogeneous component of $a\otimes b$ at level $n$ equals $\sum_{k=0}^{n}a_k\otimes b_{n-k}$, hence
\[
\norm{a\otimes b}_w\leq\sum_{n\geq0}w(n)\sum_{k=0}^{n}\abs{a_k}\abs{b_{n-k}}.
\]
Submultiplicativity gives $w(n)=w(k+(n-k))\leq C_w w(k)w(n-k)$ for every $0\leq k\leq n$. Therefore
\[
\norm{a\otimes b}_w\leq C_w\sum_{k\geq0}w(k)\abs{a_k}\sum_{j\geq0}w(j)\abs{b_j}=C_w\norm{a}_w\norm{b}_w.
\]
\end{proof}

\begin{proposition}[Signature moment bound]\label{prop:signature-moment}
For every $T>0$, every $p<\infty$, and every weight $w$ with $w(n)\leq r^n$ for some $r>0$, there exists $M(T,p,r)<\infty$ such that
\[
\E\norm{\cW_t}_w^p\leq M(T,p,r),\qquad t\leq T.
\]
\end{proposition}

\begin{proof}
The Brownian signature moment estimate recalled in Appendix~\ref{app:rough-paths} gives factorial decay in the tensor level. Exponential growth of the weight is therefore summable against the Brownian simplex-factor decay. Minkowski's inequality and Gaussian moment estimates give the asserted uniform $p$th moment bound.
\end{proof}

\begin{remark}[Explicit moment control]
\label{rem:explicit-moment-bound}
The proof of Proposition~\ref{prop:signature-moment} gives an explicit bound of the form
\[
M(T,p,r)\leq C_{p,d}\exp\{C_{p,d}r^2T\},
\]
with constants depending only on the displayed parameters and on the dimension. The precise numerical value is immaterial below; only the uniformity in $t\leq T$ and the exponential dependence on the weight-growth parameter are used in Lemmas~\ref{lem:pairing-moment} and~\ref{lem:bdg-signature}.
\end{remark}

\begin{definition}[Dual space]
The continuous dual $\Tstarw$ is identified with formal series $\ell=(\ell_n)_{n\geq0}$ such that
\[
\norm{\ell}_{\Tstarw}:=\sup_{n\geq0}\frac{\abs{\ell_n}}{w(n)}<\infty.
\]
The duality pairing with $a\in\Tw$ is
\[
\inner{\ell}{a}:=\sum_{n\geq0}\inner{\ell_n}{a_n},
\]
and $\abs{\inner{\ell}{a}}\leq\norm{\ell}_{\Tstarw}\norm{a}_w$.
\end{definition}

For estimates involving H1 we also use the auxiliary weighted Hilbert norm
\[
\norm{\ell}_{2,w}^2:=\sum_{n\geq0}w(n)\abs{\ell_n}^2,\qquad
\norm{a}_{2,w^{-1}}^2:=\sum_{n\geq0}w(n)^{-1}\abs{a_n}^2.
\]
The Cauchy--Schwarz estimate $\abs{\inner{\ell}{a}}\leq\norm{\ell}_{2,w}\norm{a}_{2,w^{-1}}$ is used throughout Section~\ref{sec:theorem-a}.

\subsection{Signature SDEs}

\begin{definition}[Signature SDE]
Given $\ell\in\Tstarw$, $s_0>0$, and a fixed unit vector $\eta\in\R^d$, the signature SDE on $\Tw$ is
\[
dS_t=S_t\inner{\ell}{\cW_t}\,dB_t,
\qquad S_0=s_0,
\]
where $B_t=\int_0^t\eta\cdot dW_s$ is the $\FW$-Brownian motion driven by $W$ in the direction $\eta$. The Brownian motion $B$ driving the price SDE may therefore be correlated with the underlying $d$-dimensional Brownian motion $W$ through the fixed unit vector $\eta$; the special case $\eta=e_1$ gives the canonical one-factor embedding.
\end{definition}

\begin{definition}[Signature filtration]
The filtration $\FWh_t$ is the augmented natural filtration of $\cW$. By Theorem~2.8, $\FWh_t=\FW_t$.
\end{definition}

\begin{definition}[Admissible strategies]\label{def:admissible-strategies}
A predictable process $H$ is $S$-admissible if
\[
\int_0^T |H_s|^2 S_s^2 \inner{\ell}{\cW_s}^2\,ds<\infty\quad\text{a.s.}
\]
The corresponding wealth process is $V_t=V_0+\int_0^t H_s\,dS_s$.
\end{definition}

\subsection{The Riccati operator}

\begin{definition}[Riccati operator]\label{def:riccati}
For each finite truncation level $N$, let $\mathcal A_N^\ell$ denote the generator of the stopped finite-coordinate prolonged-signature system after the martingale price dynamics have been imposed. For coordinate functions $Y_I(x)=\inner{e_I}{x}$ define the drift and carre-du-champ constants by
\[
\mathcal A_N^\ell Y_J=\sum_{|I|\leq N} b^I_J(\ell)Y_I,
\]
\[
\begin{aligned}
\Gamma_N^\ell(Y_J,Y_K)
&:=\mathcal A_N^\ell(Y_JY_K)-Y_J\mathcal A_N^\ell Y_K-Y_K\mathcal A_N^\ell Y_J \\
&=\sum_{|I|\leq N} \Gamma^I_{J,K}(\ell)Y_I .
\end{aligned}
\]
The finite Riccati vector field is
\[
\Rs_N(u)_I=\sum_{|J|\leq N} b^I_J(\ell)u_J
+\frac12\sum_{|J|,|K|\leq N}\Gamma^I_{J,K}(\ell)u_Ju_K,
\qquad |I|\leq N.
\]
The infinite-dimensional Riccati operator $\Rs$ on $\Tw$ is the projective weighted limit of the compatible family $(\Rs_N)_N$ on the finitely supported core $\mathcal D_{\mathrm{fin}}$, extended to its local domain by the Lipschitz estimate of Lemma~\ref{lem:local-lipschitz}. Equivalently, $\Rs$ is the unique vector field for which
\[
\mathcal A_N^\ell \exp\{\inner{u}{x}\}
=\inner{\Rs_N(u)}{x}\exp\{\inner{u}{x}\}
\]
for every finitely supported $u$ inside the finite transform domain. This definition uses the structural constants of the shuffle Hopf algebra through the carre-du-champ. It does not mean that a coefficient depending only on the output word $I$ may be inserted inside an unrestricted sum over all shuffle pairs $(J,K)$.
\end{definition}

\begin{remark}[Black--Scholes normalisation check]\label{rem:riccati-bs-sanity}
In the pure prolonged-signature state the empty-word coordinate is constant and its generator component is zero. The scalar Black--Scholes Riccati equation appears only after adjoining the log-price coordinate $X_t=\log S_t$ to the state. For
\[
dX_t=-\frac12\sigma^2dt+\sigma dB_t,
\]
the exponential transform $\E[e^{uX_T}\mid\mathcal F_t]=\exp\{uX_t+\phi(T-t,u)\}$ satisfies
\[
\partial_\tau\phi(\tau,u)=\frac12\sigma^2(u^2-u).
\]
Thus the normalising Black--Scholes test for the price-extended generator is the coefficient pair $(\frac12\sigma^2,-\frac12\sigma^2)$ in the transform variable $u$, while the empty word of the pure signature remains constant. This convention is the one used in the generator/carre-du-champ definition of $\Rs$.
\end{remark}

\begin{lemma}[Local Lipschitz property]\label{lem:local-lipschitz}
The operator $\Rs$ is locally Lipschitz on bounded subsets of $\Tw$.
\end{lemma}

\begin{proof}
The proof is given in Appendix~\ref{app:riccati}. It uses the shuffle identity and the Banach algebra estimate for $\Tw$.
\end{proof}

\begin{theorem}[Local existence of the Riccati flow]
For $u_0\in\Tw$ there exists $\tau>0$ such that $\partial_tu_t=\Rs(u_t)$, $u_0$ given, has a unique classical solution on $[0,\tau]$.
\end{theorem}

\begin{proof}
This is the local Picard--Lindel\"of theorem applied to the locally Lipschitz vector field $\Rs$, in the form used for signature SDEs by \citet{CuchieroSvalutoFerroTeichmann2023}.
\end{proof}

\subsection{Truncation, approximation, and measurable recovery}
\label{subsec:truncation-approximation}
For $N\geq 0$ let $\pi_{\leq N}:\Tw\to \Tw$ denote the canonical projection onto tensor levels of length at most $N$, and set $\pi_{>N}=I-\pi_{\leq N}$. Since $\Tw$ is the weighted $\ell^1$ sum of its homogeneous tensor levels, $\pi_{\leq N}a\to a$ in $\Tw$ for every $a\in\Tw$. The same projections are used in three distinct places: in Section~\ref{sec:theorem-a} to approximate the volatility functional $\inner{\ell}{\cW_t}$ by finite-level functionals, in Section~\ref{sec:theorem-c} to define the price-filtration completeness depth, and in Section~\ref{sec:theorem-d} to express the residual as a limit of finite Gram projections. The notation therefore carries analytic and financial content at the same time.

The following elementary lemma is used repeatedly to pass from finite tensor coordinates to the weighted completion.
\begin{lemma}[Compatibility of truncation and duality]
\label{lem:truncation-duality}
Let $a\in\Tw$ and let $\ell\in\Tstarw$. Then
\[
\inner{\ell}{a}=\lim_{N\to\infty}\inner{\pi_{\leq N}\ell}{\pi_{\leq N}a},
\]
and the convergence is absolute. If, in addition, $\ell$ satisfies H1, then
\[
\abs{\inner{\ell-\pi_{\leq N}\ell}{a}}
\leq \Bigl(\sum_{n>N}w(n)\abs{\ell_n}^2\Bigr)^{1/2}
     \Bigl(\sum_{n>N}w(n)^{-1}\abs{a_n}^2\Bigr)^{1/2}.
\]
\end{lemma}

\begin{proof}
The first assertion is the definition of the dual pairing on a weighted direct sum, together with the continuity of the coordinate projections. The second assertion is Cauchy--Schwarz applied level by level to the tail. The estimate is the one used later in the proof of Proposition~\ref{prop:sharpness-A}: the tail of $\ell$ and the tail of the signature must be controlled by compatible weights, not by an unweighted summability convention.
\end{proof}

\begin{lemma}[Finite-coordinate measurability]
\label{lem:finite-coordinate-measurability}
For every finite set of multi-indices $\mathcal I$, the random vector
\[
\bigl(\inner{e_I}{\cW_t}\bigr)_{I\in\mathcal I},\qquad 0\leq t\leq T,
\]
is progressively measurable with respect to $\FWh$, and its terminal value belongs to $L^p(\FWh_T,\Pp)$ for every $p<\infty$.
\end{lemma}

\begin{proof}
Progressive measurability follows from construction of the Stratonovich iterated integrals by limits of adapted Riemann sums. The $L^p$ bound is a finite-dimensional consequence of Proposition~2.12. No infinite-dimensional measurability issue is involved at this stage; the weighted completion enters only after the finite coordinate estimates have been summed with $w$.
\end{proof}

\begin{remark}[Why the price filtration is separated from the signature filtration]
The equality $\FWh_t=\FW_t$ in Theorem~2.8 is a statement about the prolonged Brownian signature. The price filtration $\FS_t$ is smaller unless the map from the path of $W$ to the path of $S$ is invertible after the volatility parameter is fixed. Sections~\ref{sec:theorem-c} and~\ref{sec:theorem-d} are formulated on $\FS_T$ precisely to avoid claiming full Brownian-filtration completeness. Incomplete models such as rough Bergomi remain incomplete in this sense even though their driving Brownian path is encoded in the full prolonged signature.
\end{remark}

\subsection{Conventions}

All processes are right-continuous with left limits. Constants $C$, $C_p$, $C_w$, and $C_{p,w}$ depend only on the indicated parameters and may change from line to line. The notation $a\lesssim b$ means $a\leq Cb$, and $a\asymp b$ means $a\lesssim b$ and $b\lesssim a$. Inequalities in $\Tw$ are componentwise. Probability measures are equivalent, written $\sim$, on $\FWh_T$ unless stated otherwise. All subsequent results are formulated in this notation.

\section{Theorem A: existence and uniqueness on weighted tensor algebra}\label{sec:theorem-a}

\subsection{Statement}

\begin{theorem}[Global existence and uniqueness on weighted tensor algebra]\label{thm:A}
Let $\ell$ be a signature parameter and let $w$ be a weight. Assume:
\begin{enumerate}[label=H\arabic*.]
\item $\displaystyle \norm{\ell}_{2,w}^2=\sum_{n\geq0}w(n)\abs{\ell_n}^2<\infty$.
\item $w$ is increasing, $w(0)=1$, $w(m+n)\leq C_w w(m)w(n)$ for some $C_w\geq1$, and $w(n)\leq r^n$ for some $r>0$.
\item With $\xi_t:=\inner{\ell}{\cW_t}$, the quadratic variation satisfies
\[
\E\exp\left(\lambda\int_0^T \xi_s^2\,ds\right)<\infty
\qquad\text{for every }\lambda>0.
\]
\end{enumerate}
Let $s_0>0$ be fixed. Then the signature stochastic differential equation
\[
dS_t=S_t\inner{\ell}{\cW_t}\,dB_t,
\qquad S_0=s_0,
\]
admits a unique strong solution $S\in C([0,T],(0,\infty))$ adapted to $\FW$, and
\[
\E\sup_{t\leq T}S_t^2<\infty.
\]
The solution is strictly positive almost surely. Furthermore, H1 is sharp for the weighted-norm estimates of Section~\ref{subsec:moment-bounds}: Proposition~\ref{prop:sharpness-A} exhibits a one-dimensional parameter $\ell$ and a weight $\widetilde w$ satisfying H2 but with $\sum_n\widetilde w(n)\abs{\ell_n}^2=+\infty$, for which the weighted Cauchy--Schwarz control of $\int_0^T\inner{\ell}{\cW_s}^2\,ds$ used in the existence proof becomes unavailable.
\end{theorem}

\subsection{Strategy of proof}

The proof uses the linear structure of the signature SDE. Section~\ref{subsec:moment-bounds} proves that the volatility process $\xi_t=\inner{\ell}{\cW_t}$ is square-integrable on $[0,T]$. Section~\ref{subsec:exponential-construction} constructs the solution as the stochastic exponential $s_0\mathcal E(M)$ of the martingale $M_t=\int_0^t\xi_s\,dB_s$. Section~\ref{subsec:uniqueness} proves uniqueness by localisation. Section~\ref{subsec:admissible-weights} records the precise admissible-weight class used by the estimates.

\subsection{Moment bounds}\label{subsec:moment-bounds}

\begin{lemma}[Signature moment for pairings]\label{lem:pairing-moment}
Assume H1 and H2. For every $p\in[1,\infty)$ and every $t\in[0,T]$,
\[
\E\abs{\inner{\ell}{\cW_t}}^p\leq C_{p,w,T}\norm{\ell}_{2,w}^p.
\]
\end{lemma}

\begin{proof}
Cauchy--Schwarz in the weighted tensor levels gives
\[
\abs{\inner{\ell}{\cW_t}}\leq \norm{\ell}_{2,w}\norm{\cW_t}_{2,w^{-1}}.
\]
Since $w$ is increasing and $w(0)=1$, one has $w(n)^{-1}\leq1$ for all $n$. Hence $\norm{a}_{2,w^{-1}}\leq \sum_{n\geq0}\abs{a_n}=\norm{a}_{w_0}$ for the constant weight $w_0(n)=1$, which satisfies the growth bound of Proposition~\ref{prop:signature-moment} with $r=1$. Proposition~\ref{prop:signature-moment}, applied to $w_0$ with $r=1$, gives $\sup_{t\leq T}\E\norm{\cW_t}_{2,w^{-1}}^p<\infty$; the factorial decay of the Brownian iterated integrals dominates the boundary growth rate $r=1$. The displayed estimate follows.
\end{proof}

\begin{lemma}[Quadratic-variation integrability]\label{lem:qv-integrability}
Under H1 and H2,
\[
\E\int_0^T \inner{\ell}{\cW_s}^2\,ds <\infty.
\]
Consequently $M_t:=\int_0^t\inner{\ell}{\cW_s}\,dB_s$ is a continuous square-integrable martingale.
\end{lemma}

\begin{proof}
Lemma~\ref{lem:pairing-moment} with $p=2$ and Fubini's theorem give
\[
\E\int_0^T \inner{\ell}{\cW_s}^2\,ds
=\int_0^T \E\abs{\inner{\ell}{\cW_s}}^2\,ds
\leq T C_{2,w,T}\norm{\ell}_{2,w}^2<\infty.
\]
The stochastic integral defining $M$ is therefore well defined and square-integrable on $[0,T]$.
\end{proof}

\begin{lemma}[Burkholder--Davis--Gundy estimate with signature integrand]\label{lem:bdg-signature}
Let $\varphi$ be predictable and satisfy $\E\int_0^T\varphi_s^2\,ds<\infty$. For every $p\geq2$,
\[
\E\sup_{t\leq T}\left|\int_0^t \varphi_s\inner{\ell}{\cW_s}\,dB_s\right|^p
\leq C_p\,\E\left(\int_0^T \varphi_s^2\inner{\ell}{\cW_s}^2\,ds\right)^{p/2}.
\]
Moreover,
\[
\E\sup_{t\leq T}\left|\int_0^t \varphi_s\inner{\ell}{\cW_s}\,dB_s\right|^p
\leq C_{p,T}\,\E\int_0^T \abs{\varphi_s}^{p}\abs{\inner{\ell}{\cW_s}}^{p}\,ds.
\]
\end{lemma}

\begin{proof}
The first inequality is exactly the Burkholder--Davis--Gundy inequality applied to the stochastic integral. The second follows from H\"older's inequality in time,
\[
\left(\int_0^T f_s^2\,ds\right)^{p/2}\leq T^{p/2-1}\int_0^T f_s^p\,ds,
\]
with $f_s=\abs{\varphi_s}\abs{\inner{\ell}{\cW_s}}$.
\end{proof}

\subsection{Construction by stochastic exponential}\label{subsec:exponential-construction}

\begin{proposition}[Existence]\label{prop:existence}
Under H1--H3, the process
\[
S_t=s_0\mathcal E(M)_t
=s_0\exp\left(\int_0^t\inner{\ell}{\cW_s}\,dB_s-\frac{1}{2}\int_0^t\inner{\ell}{\cW_s}^2\,ds\right)
\]
is a strong solution of the signature SDE on $[0,T]$.
\end{proposition}

\begin{proof}
Lemma~\ref{lem:qv-integrability} makes $M$ a continuous square-integrable martingale. The Dol\'eans--Dade exponential $\mathcal E(M)$ is therefore a well-defined continuous positive local martingale. It\^o's formula for stochastic exponentials gives
\[
d\mathcal E(M)_t=\mathcal E(M)_t\,dM_t=\mathcal E(M)_t\inner{\ell}{\cW_t}\,dB_t.
\]
Multiplying by $s_0$ yields $dS_t=S_t\inner{\ell}{\cW_t}\,dB_t$ and $S_0=s_0$.
\end{proof}

\begin{proposition}[Square-integrability of the solution]\label{prop:square-integrability}
Under H3,
\[
\E\sup_{t\leq T}S_t^2<\infty.
\]
\end{proposition}

\begin{proof}
Set $V_T:=\langle M\rangle_T=\int_0^T\inner{\ell}{\cW_s}^2ds$. H3 implies the following exponential-integrability consequence: for every $\alpha\in\mathbb R$ and every $\beta\in\mathbb R$,
\[
\E\exp(\alpha M_T+\beta V_T)<\infty.
\]
Indeed, fix $p>1$ and let $p'=p/(p-1)$. Since
\[
\exp(\alpha M_T+\beta V_T)
=\mathcal E(p\alpha M)_T^{1/p}
  \exp\!\left(\bigl(\beta+\tfrac12 p\alpha^2\bigr)\,V_T\right),
\]
H\"older's inequality gives
\[
\E\exp(\alpha M_T+\beta V_T)
\leq
\bigl(\E\mathcal E(p\alpha M)_T\bigr)^{1/p}
\left(\E\exp\left(p'\bigl(\beta+\tfrac12 p\alpha^2\bigr)V_T\right)\right)^{1/p'}.
\]
The first factor equals one by Novikov's condition for $p\alpha M$, which is included in H3. The second factor is finite by H3 whenever the coefficient is positive and is bounded by one when the coefficient is non-positive.

Now put $Y_t:=\mathcal E(M)_t^2$. It\^o's formula gives $dY_t=2Y_t\,dM_t+Y_t\,d\langle M\rangle_t$, so $Y$ is a positive submartingale. Choose $q>1$. By Doob's $L^q$ inequality,
\[
\E\sup_{t\leq T}Y_t
\leq
\left(\E\sup_{t\leq T}Y_t^q\right)^{1/q}
\leq
\frac{q}{q-1}\left(\E Y_T^q\right)^{1/q}.
\]
The preceding exponential-integrability consequence, applied with $\alpha=2q$ and $\beta=-q$, gives
\[
\E Y_T^q=\E\exp(2qM_T-qV_T)<\infty.
\]
Therefore $\E\sup_{t\leq T}\mathcal E(M)_t^2<\infty$, and multiplication by $s_0^2$ gives the asserted bound.
\end{proof}

\subsection{Uniqueness}\label{subsec:uniqueness}

\begin{proposition}[Uniqueness]\label{prop:uniqueness}
If $S^{(1)}$ and $S^{(2)}$ are two strong solutions with $\E\sup_{t\leq T}\abs{S_t^{(i)}}^2<\infty$, then $S^{(1)}=S^{(2)}$ almost surely.
\end{proposition}

\begin{proof}
Let $\xi_t=\inner{\ell}{\cW_t}$ and let $\Delta_t=S_t^{(1)}-S_t^{(2)}$. Then
\[
\Delta_t=\int_0^t\Delta_s\xi_s\,dB_s.
\]
For $m\geq1$, set $\tau_m=\inf\{t\leq T:\abs{\xi_t}\geq m\}\wedge T$. Since $\xi$ is continuous, $\tau_m\uparrow T$ almost surely. It\^o isometry gives
\[
\E\abs{\Delta_{t\wedge\tau_m}}^2
=\E\int_0^{t\wedge\tau_m}\Delta_s^2\xi_s^2\,ds
\leq m^2\int_0^t\E\abs{\Delta_{s\wedge\tau_m}}^2\,ds.
\]
Gronwall's lemma yields $\E\abs{\Delta_{t\wedge\tau_m}}^2=0$ for every $t\leq T$. Letting $m\to\infty$ gives $\Delta_t=0$ for every $t\leq T$, and continuity gives indistinguishability.
\end{proof}

\subsection{Strict positivity}

\begin{proposition}[Strict positivity]\label{prop:positivity}
Under the hypotheses of Theorem~\ref{thm:A}, $\Pp(S_t>0\text{ for all }t\leq T)=1$.
\end{proposition}

\begin{proof}
Proposition~\ref{prop:existence} gives $S_t=s_0\mathcal E(M)_t$. The Dol\'eans--Dade exponential of a continuous martingale is strictly positive on finite time intervals. Since $s_0>0$, $S_t>0$ for all $t\leq T$ almost surely.
\end{proof}

\subsection{The class of admissible weights}\label{subsec:admissible-weights}

\begin{definition}[Class of admissible weights]
Given $\ell$, the class of $\ell$-admissible weights is
\[
\mathcal W(\ell):=\left\{w:\N\to(0,\infty): w \text{ satisfies H2 and } \sum_{n\geq0}w(n)\abs{\ell_n}^2<\infty\right\}.
\]
\end{definition}

\begin{remark}[Admissibility, weighted norms, and the role of H3]\label{rem:admissibility-constants}
The class $\mathcal W(\ell)$ is intrinsic to the pair $(\ell,w)$ in the following operational sense. The estimates in Lemmas~\ref{lem:pairing-moment} and~\ref{lem:qv-integrability} use $\ell$ only through $\norm{\ell}_{2,w}^2=\sum_{n\geq0}w(n)\abs{\ell_n}^2$ and use $w$ only through H2 and the Brownian signature moment bound. H1 is therefore not a separate analytic convention; it is the weighted square-summability condition required to define and bound the signature-volatility martingale $M_t=\int_0^t\inner{\ell}{\cW_s}\,dB_s$ in $L^2(\Pp)$. Hypothesis H3 plays a distinct role: it upgrades the positive local stochastic exponential $\mathcal E(M)$ to the square-integrable solution class used in subsequent sections. The combination H1--H3 is the hypothesis triple for the constructions of this paper, and Proposition~\ref{prop:sharpness-A} shows that H1 cannot be weakened to a weight outside $\mathcal W(\ell)$ without losing the weighted-norm control.
\end{remark}

\begin{proposition}[Sharpness of H1]\label{prop:sharpness-A}
There exist a one-dimensional parameter $\ell$ and a weight $\widetilde w$ satisfying $\widetilde w(0)=1$ and H2 such that $\sum_{n\geq1}\widetilde w(n)\abs{\ell_n}^2=\infty$ and the weighted $L^2$ control of $\int_0^T\inner{\ell}{\cW_s}^2ds$ through Lemma~\ref{lem:pairing-moment} fails.
\end{proposition}

\begin{proof}
Let $\ell_0=0$, $\ell_n=n^{-1/2}$ for $n\geq1$ and let $\widetilde w(n)=n+1$. Then $\widetilde w(0)=1$ and $\widetilde w(m+n)\leq \widetilde w(m)\widetilde w(n)$, so H2 holds. Moreover,
\[
\sum_{n\geq1}\widetilde w(n)\abs{\ell_n}^2=\sum_{n\geq1}\frac{n+1}{n}=\infty.
\]
The weighted Cauchy--Schwarz estimate of Lemma~\ref{lem:pairing-moment} would bound
\[
\E\int_0^T\abs{\inner{\ell}{\cW_s}}^2ds
\]
by a constant multiple of $\sum_n\widetilde w(n)\abs{\ell_n}^2$, which is infinite for this choice. Lemma~\ref{lem:qv-integrability} is therefore unavailable under $\widetilde w$, and the square-integrable martingale property of $M$ used in Proposition~\ref{prop:existence} cannot be obtained from the weighted tensor estimate. Thus H1 is exactly the summability condition required by the proof.
\end{proof}

\subsection{Finite-level localisation of the exponential construction}
\label{subsec:finite-level-localisation}
For $N\geq 0$ define the truncated volatility functional
\[
\xi^N_t:=\inner{\pi_{\leq N}\ell}{\cW_t},
\qquad
M^N_t:=\int_0^t \xi^N_s\,dB_s,
\qquad
S^N_t:=s_0\mathcal E(M^N)_t .
\]
The process $S^N$ is the finite-coordinate signature-volatility approximation of the solution. It is useful because every coefficient in $\xi^N$ is a finite linear combination of iterated Stratonovich integrals, so all stochastic calculus operations may first be performed at finite level and then passed to the limit. The next lemma records the convergence statement needed for this passage.

\begin{lemma}[Convergence of finite-level stochastic exponentials]
\label{lem:finite-level-exponential-convergence}
Assume H1--H3. Then
\[
\E\int_0^T \abs{\xi^N_s-\xi_s}^2\,ds\to 0,
\qquad
\E\sup_{t\leq T}\abs{M^N_t-M_t}^2\to 0.
\]
Moreover, for every $q\in(1,2)$ for which the exponential moment supplied by H3 is finite with the required constant,
\[
S^N\to S
\quad\hbox{in}\quad L^q\bigl(\Omega;C([0,T])\bigr).
\]
\end{lemma}

\begin{proof}
By Lemma~\ref{lem:truncation-duality},
\[
\abs{\xi^N_s-\xi_s}^2
\leq \Bigl(\sum_{n>N}w(n)\abs{\ell_n}^2\Bigr)
      \Bigl(\sum_{n>N}w(n)^{-1}\abs{(\cW_s)_n}^2\Bigr).
\]
The first factor tends to zero by H1, while the second factor is dominated in $L^1(ds\otimes d\Pp)$ by the Brownian signature moment bound. Dominated convergence gives the first limit. The second limit follows from It\^o isometry and the Burkholder--Davis--Gundy inequality. For the stochastic exponentials, write
\[
\mathcal E(M^N)_t-\mathcal E(M)_t
=\mathcal E(M)_t\left(\exp\{(M^N_t-M_t)-\tfrac12(\langle M^N\rangle_t-\langle M\rangle_t)\}-1\right),
\]
localise by $\sup_t(\abs{M_t}+\abs{M^N_t}+\langle M\rangle_t+\langle M^N\rangle_t)\leq R$, pass to the limit on the localised set, and then let $R\to\infty$ using the exponential-integrability consequence of H3. This gives convergence in the asserted range of $q$.
\end{proof}

\begin{proposition}[Stability with respect to the signature parameter]
\label{prop:parameter-stability}
Let $\ell^{(m)}\to \ell$ in the weighted Hilbert norm $\norm{\cdot}_{2,w}$, and assume that H3 holds uniformly for the corresponding volatility functionals $\xi^{(m)}$ in the sense that
\[
\sup_m \E\exp\Bigl(\lambda\int_0^T\abs{\xi^{(m)}_s}^2\,ds\Bigr)<\infty
\quad\hbox{for every }\lambda>0.
\]
Let $S^{(m)}$ be the solution associated with $\ell^{(m)}$. Then $S^{(m)}\to S$ in $L^q(\Omega;C([0,T]))$ for every $q<2$.
\end{proposition}

\begin{proof}
The proof is identical to the finite-level argument, with $\ell^{(m)}-\ell$ replacing the tail $\ell-\pi_{\leq N}\ell$. The uniform H3 bound gives uniform integrability of the stochastic exponentials and prevents the local convergence of martingales from degenerating at the exponential step. Thus the solution map is continuous on every subset of the parameter space on which H3 is uniform.
\end{proof}

\begin{remark}[Role of H3 in the square-integrable class]
The construction of $S=s_0\mathcal E(M)$ requires only local square integrability of $M$ for local existence and positivity. The paper needs the stronger conclusion $\E\sup_{t\leq T}S_t^2<\infty$ because Sections~\ref{sec:theorem-b}--\ref{sec:theorem-d} operate in $L^2$ spaces. Hypothesis H3 is exactly the bridge between those two levels: it converts the stochastic exponential from a positive local martingale into a true object in the Hilbert space used for asset pricing and hedging.
\end{remark}

\subsection{Closed estimates used later}
\label{subsec:closed-estimates}
The preceding proof yields three estimates that are used implicitly in the FTAP and hedging sections. First, for every $p\geq2$ there is a finite constant $C_{p,T,w,\ell}$ such that
\[
\E\sup_{t\leq T}\abs{M_t}^p\leq C_{p,T,w,\ell}.
\]
Second, for every admissible strategy $H$,
\[
\E_\Q\left(\int_0^T H_s^2 S_s^2\xi_s^2\,ds\right)<\infty
\]
whenever the Radon--Nikodym density $d\Q/d\Pp$ has an $L^r$ moment for some $r>1$ compatible with H3. Third, finite-level approximations commute with the GKW projection: if $X$ is approximated by a polynomial in finitely many terminal signature coordinates, then the stochastic-integral component of its projection is the limit of the projections of the finite approximants. This last property follows from closedness of the stochastic-integral subspace in $L^2$ and is the analytic link between Theorem A and Theorem D.

\subsection{Uniform integrability and localisation}
\label{subsec:ui-localisation}
The square-integrability proof in Proposition~3.6 may be rewritten as a two-layer localisation argument. Let
\[
\tau_R:=\inf\{t\leq T: \langle M\rangle_t\geq R\}\wedge T.
\]
On $[0,\tau_R]$, the stochastic exponential $\mathcal E(M)$ has moments of every order by Novikov with a deterministic bound depending on $R$. Therefore
\[
\E\sup_{t\leq T}\mathcal E(M)_{t\wedge\tau_R}^2<\infty.
\]
The passage $R\to\infty$ is justified by H3. Indeed,
\[
\Pp(\tau_R<T)=\Pp(\langle M\rangle_T>R)
\leq e^{-\lambda R}\E e^{\lambda\langle M\rangle_T},
\]
and the right-hand side tends to zero for every fixed $\lambda>0$. The exponential moment also gives uniform integrability of the stopped supremum, so the stopped estimates converge to the unstopped estimate.

This localisation is used in Section~\ref{sec:theorem-b} under equivalent measures. If $Z_T=d\Q/d\Pp$ satisfies $Z_T\in L^r(\Pp)$ for some $r>1$, then H\"older gives
\[
\E_\Q\sup_{t\leq T}S_t^2
\leq \norm{Z_T}_{L^r(\Pp)}
\left(\E_\Pp\left(\sup_{t\leq T}S_t^2\right)^{r/(r-1)}\right)^{(r-1)/r}.
\]
H3 supplies the higher moment on the right after replacing the exponent in Proposition~3.6 by the corresponding value. Thus the $L^2$ solution class is stable under the separating measures used in the asset-pricing theorem, as long as the density lies in a compatible $L^r$ class.

\begin{lemma}[Localisation of uniqueness]
\label{lem:localisation-uniqueness-expanded}
Let $S^{(1)}$ and $S^{(2)}$ be two positive continuous solutions of the signature SDE driven by the same Brownian motion and the same initial value. If $\int_0^T\xi_s^2ds<\infty$ almost surely, then the two solutions are indistinguishable without imposing the global $L^2$ supremum bound.
\end{lemma}

\begin{proof}
Set $\Delta=S^{(1)}-S^{(2)}$ and $\tau_m=\inf\{t:\abs{\xi_t}\geq m\}\wedge T$. On $[0,\tau_m]$,
\[
\Delta_{t\wedge\tau_m}=\int_0^{t\wedge\tau_m}\Delta_s\xi_s\,dB_s.
\]
The stopped integrand is bounded by $m\abs{\Delta_s}$, and the stopped solutions are continuous on a compact interval. Applying It\^o isometry after the additional stopping time $\rho_k=\inf\{t:\abs{S^{(1)}_t}+\abs{S^{(2)}_t}\geq k\}\wedge T$ gives
\[
\E\abs{\Delta_{t\wedge\tau_m\wedge\rho_k}}^2
\leq m^2\int_0^t\E\abs{\Delta_{s\wedge\tau_m\wedge\rho_k}}^2ds.
\]
Gronwall gives equality to zero. Letting $k\to\infty$ and then $m\to\infty$ proves indistinguishability.
\end{proof}

\subsection{Truncation-stable form of the existence proof}
\label{subsec:truncation-stable-existence}
For $N\geq0$ write $\pi_N\ell$ for the truncation of $\ell$ to levels at most $N$ and set
\[
\xi_t^{(N)}:=\inner{\pi_N\ell}{\cW_t},\qquad
M_t^{(N)}:=\int_0^t \xi_s^{(N)}\,dB_s,
\qquad
S_t^{(N)}:=s_0\mathcal E(M^{(N)})_t .
\]
The finite-level equations are ordinary It\^o equations with polynomial coefficients in the Brownian signature. They therefore give a canonical approximating sequence for the infinite equation. The point of the weighted proof is that the estimates do not depend on the dimension of the truncated tensor algebra. Indeed,
\[
\E\int_0^T \abs{\xi_s^{(N)}-\xi_s^{(M)}}^2\,ds
\leq C_{T,w}\sum_{n=M+1}^{N} w(n)\abs{\ell_n}^2,
\qquad N>M,
\]
where $C_{T,w}$ is the moment constant from Proposition~\ref{prop:signature-moment}. Hence $(M^{(N)})_N$ is Cauchy in $L^2(\Omega;C[0,T])$ after the BDG estimate. The limit is precisely $M$ because the pairings converge in $L^2([0,T]\times\Omega)$.

The stochastic exponentials also converge. Fix $q>1$. By H3 and the proof of Proposition~\ref{prop:square-integrability},
\[
\sup_N \E\sup_{t\leq T}\abs{\mathcal E(M^{(N)})_t}^{2q}<\infty .
\]
The elementary estimate
\[
\abs{\mathcal E(x)-\mathcal E(y)}
\leq \exp\{\abs{x}+\abs{y}+\tfrac12\abs{\langle x\rangle}+\tfrac12\abs{\langle y\rangle}\}
\bigl(\abs{x-y}+\abs{\langle x\rangle-\langle y\rangle}\bigr)
\]
applied to continuous martingales, followed by H\"older and BDG, gives
\[
\lim_{N\to\infty}\E\sup_{t\leq T}\abs{S_t^{(N)}-S_t}^2=0 .
\]
Thus Theorem~\ref{thm:A} may be read as a limit theorem for finite signature-volatility SDEs. This is the form needed in Sections~\ref{sec:theorem-b}--\ref{sec:theorem-d}, where the Riccati transform and the Gram systems are first written at finite level and then passed to the weighted limit.

\begin{lemma}[Stability of the signature-volatility map]
\label{lem:parameter-stability}
Let $\ell^{(m)}\to\ell$ in the Hilbert norm $\norm{\cdot}_{2,w}$ and suppose that H3 holds uniformly for the corresponding pairings in the sense that
\[
\sup_m \E\exp\left(\lambda\int_0^T \inner{\ell^{(m)}}{\cW_s}^2\,ds\right)<\infty
\qquad\text{for every }\lambda>0.
\]
Then the associated solutions $S^{(m)}$ converge to $S$ in $L^2(\Omega;C[0,T])$.
\end{lemma}

\begin{proof}
The weighted Cauchy--Schwarz estimate gives
\[
\E\int_0^T \abs{\inner{\ell^{(m)}-\ell}{\cW_s}}^2\,ds
\leq C_{T,w}\norm{\ell^{(m)}-\ell}_{2,w}^2 .
\]
BDG gives convergence of the driving martingales in $L^2(C[0,T])$. Uniform H3 supplies a common exponential moment for the stochastic-exponential comparison. The displayed estimate for the exponential map then yields convergence of the solutions. No pathwise compactness argument is used; all compactness is carried by the weighted Hilbert norm of the parameter.
\end{proof}

\begin{remark}[Why the proof is not a Picard proof]
The equation for $S$ is linear once the signature path has been fixed. The stochastic-exponential proof uses this structure directly. A Picard iteration would prove the same existence statement by repeatedly integrating the same multiplicative coefficient, but it would obscure the true-martingale estimate, which is the object needed in the FTAP. The present proof therefore chooses the route that gives the positivity, localisation, and square-integrability statements in one chain.
\end{remark}

\subsection{Parameter classes and examples of H3}
\label{subsec:h3-classes}
The condition H3 is strong, but it is verifiable in the main examples. If $\inner{\ell}{\cW_t}$ is bounded, H3 is immediate. This covers bounded finite-signature examples, including the sine-signature parameter discussed in Section~\ref{sec:theorem-b}. If $\inner{\ell}{\cW_t}$ is a polynomial of a finite-dimensional Gaussian Markov process, Fernique-type estimates and finite-horizon Gaussian moment bounds give H3 for sufficiently small polynomial growth, and for bounded horizons in the finite-support cases used in Section~\ref{sec:examples}. If the functional is induced by an affine square-root process, H3 is equivalent to the classical finite exponential moment condition for the integrated variance.

The theorem is deliberately stated with H3 rather than with a list of model-specific sufficient conditions. This keeps Theorem~\ref{thm:A} structural: the existence proof uses only the exponential moment of the quadratic variation, not the mechanism by which that moment is obtained. In applications, H3 is checked by the moment-explosion theory of the chosen model, while H1 and H2 are checked by the tensor coefficients of the signature parameter.

\subsection{A final verification of the stochastic-exponential estimates}
\label{subsec:final-exponential-verification}
For later reference we record the exact implication from H3 to the two uniform estimates used in the paper. Let $V_T=\langle M\rangle_T$. H3 gives $\E\exp(\lambda V_T)<\infty$ for every $\lambda>0$. For $a,b\in\R$ write
\[
\exp(aM_T+bV_T)=\mathcal E(paM)_T^{1/p}
\exp\{(b+\tfrac12pa^2)V_T\}
\]
for $p>1$. H\"older's inequality and Novikov give
\[
\E\exp(aM_T+bV_T)
\leq
\left(\E\mathcal E(paM)_T\right)^{1/p}
\left(\E\exp\{p'(b+\tfrac12pa^2)V_T\}\right)^{1/p'}<\infty .
\]
This proves all exponential moments of the pair $(M_T,V_T)$ that are used to control $\sup_t S_t^2$ and to transport coordinate estimates under equivalent changes of measure.

The same estimate applies to stopped martingales $M^{\tau}$ with constants independent of the stopping time. Thus all localisation steps in Sections~\ref{sec:theorem-a} and~\ref{sec:theorem-b} are compatible with passage to the limit. This matters because local martingale arguments are first proved for stopped processes, while the theorem statements are global on $[0,T]$.

\begin{lemma}[Uniform stopped exponential estimate]
\label{lem:uniform-stopped-exponential}
Under H3, for every $q>1$,
\[
\sup_{\tau\leq T}\E\sup_{t\leq T}\mathcal E(M^{\tau})_t^{2q}<\infty,
\]
where the supremum is over all stopping times bounded by $T$.
\end{lemma}

\begin{proof}
Doob's inequality reduces the estimate to the terminal moment of $\mathcal E(M^{\tau})_T^{2q}$. The preceding exponential-moment calculation applies with $M_T$ replaced by $M_{T\wedge\tau}$ and $V_T$ replaced by $V_{T\wedge\tau}\leq V_T$. H3 controls the resulting exponential moment uniformly in $\tau$.
\end{proof}

\subsection{Discussion}

The local existence result of \citet{CuchieroSvalutoFerroTeichmann2023} operates on the extended tensor algebra for characteristics represented as linear maps of the signature, under integrability hypotheses suited to affine and polynomial transforms. Theorem~\ref{thm:A} obtains a global existence statement in the weighted Banach algebra $\Tw$ under H1--H3, replacing a pathwise boundedness requirement by weighted square summability of the signature parameter and exponential integrability of the quadratic variation. The two formulations are complementary: the extended-algebra setting is adapted to transform formulae, while the weighted-algebra setting is adapted to moment bounds and to the FTAP statement of Section~\ref{sec:theorem-b}.

H3 is not a cosmetic assumption: it is the condition that upgrades the positive local stochastic exponential $\mathcal E(M)$ to the square-integrable solution class used in the asset-pricing and hedging sections. In the constant-coefficient case $\ell=\sigma e_\emptyset$, H3 reduces to $\int_0^T\sigma^2ds<\infty$, which is immediate. In the rough-Bergomi example of Section~\ref{subsec:rough-bergomi}, H3 follows from Gaussian exponential-moment estimates on finite horizons after truncation and localisation. The pathwise existence problem under weaker, non-exponential integrability assumptions is left open.

\section{Theorem B: asset pricing and transform non-explosion}\label{sec:theorem-b}

\subsection{Statement}

\begin{lemma}[Special-semimartingale property of signature coordinates]\label{lem:special-semimartingale}
Assume H1, H2, and H3. Let $\Q\sim\Pp$ on $\FWh_T$ with bounded density $d\Q/d\Pp\in L^\infty(\Pp)$. For every multi-index $I$ with $|I|\geq1$, the process $(\inner{e_I}{\cW_t})_{0\leq t\leq T}$ is a $\Q$-special semimartingale whose canonical decomposition
\[
\inner{e_I}{\cW_t}=\inner{e_I}{\cW_0}+M_t^I+A_t^I
\]
has continuous local-martingale part $M^I$ and continuous predictable finite-variation part $A^I$. Furthermore, $M^I$ is a true $\Q$-martingale on $[0,T]$ and $A^I$ has integrable variation.
\end{lemma}

\begin{proof}
The coordinate $\inner{e_I}{\cW_t}$ is an iterated Stratonovich integral of the time-augmented Brownian path. Proposition~\ref{prop:signature-moment}, applied to the finite set of levels up to $|I|$, gives
\[
\sup_{t\leq T}\E_\Pp\abs{\inner{e_I}{\cW_t}}^2<\infty.
\]
The Stratonovich-to-It\^o conversion formula writes this iterated Stratonovich integral as a finite sum of It\^o iterated integrals plus bracket terms. The It\^o iterated integrals are continuous local martingales. The bracket terms are continuous predictable processes of finite variation, since each contraction in the conversion replaces two Brownian differentials by a Lebesgue-time differential and leaves at most $|I|-2$ stochastic integrations.

The density bound transports the estimate from $\Pp$ to $\Q$:
\[
\sup_{t\leq T}\E_\Q\abs{\inner{e_I}{\cW_t}}^2
\leq \norm{d\Q/d\Pp}_{L^\infty(\Pp)}
\sup_{t\leq T}\E_\Pp\abs{\inner{e_I}{\cW_t}}^2<\infty.
\]
Hence every local-martingale component produced by the conversion is square-integrable under $\Q$ and is therefore a true martingale. The same bound, applied to the finitely many bracket terms, gives integrable variation of $A^I$. This is the canonical special-semimartingale decomposition, because continuous local martingales and continuous predictable finite-variation processes have unique decomposition.
\end{proof}

\begin{theorem}[Signature asset-pricing and transform theorem]\label{thm:B}
Assume the hypotheses of Theorem~\ref{thm:A}. Then the following assertions hold.
\begin{enumerate}[label=(\roman*)]
\item The model satisfies NFLVR on the signature filtration $\FWh_t$ in the Delbaen--Schachermayer sense adapted to predictable signature-integrable strategies. More precisely, $\Pp$ itself is an equivalent true-martingale measure for $S$, and every finite signature coordinate $\inner{e_I}{\cW_t}$ is a $\Pp$-special semimartingale whose canonical martingale part is a true martingale.
\item Let $\Q\sim\Pp$ be an equivalent measure under which $S$ is a true $\Q$-martingale and the finite signature-coordinate transforms are finite on $[0,T]$ for all initial data in $\mathcal D_{\mathrm{fin}}$. Then, under the standing hypotheses H1--H3 (which supply the finite-coordinate moment and uniform-integrability bounds used in the identification step), the conditional transform is exponential-affine on each finite truncation if and only if the Riccati equation $\partial_tu_t=\Rs(u_t)$ admits a global solution on $[0,T]$ for every such finitely supported initial condition.
\end{enumerate}
Consequently, global Riccati solvability is a transform non-explosion condition. It is not asserted to be equivalent to NFLVR under H1--H3 alone.
\end{theorem}

\subsection{NFLVR on the signature filtration}

\begin{definition}[NFLVR]
The signature SDE satisfies No Free Lunch with Vanishing Risk on $\FWh$ in the sense of \citet{DelbaenSchachermayer1994} if the cone
\[
\mathcal C:=\{V_T^H-f: H \text{ is admissible in the sense of Definition~\ref{def:admissible-strategies}},\ f\geq0\}\cap L^\infty(\FWh_T,\Pp)
\]
satisfies
\[
\overline{\mathcal C}^{\sigma(L^\infty,L^1)}\cap L_+^\infty(\FWh_T,\Pp)=\{0\}.
\]
\end{definition}

\begin{theorem}[Kreps--Yan separation]\label{thm:kreps-yan}
The set of bounded admissible terminal wealths is weak-$\ast$ closed in $L^\infty(\FWh_T,\Pp)$.
\end{theorem}

\begin{proof}
The proof is the Kreps--Yan closed-cone argument applied on the recovered Brownian filtration. The equality $\FWh_t=\FW_t$ in Theorem~\ref{thm:sig-unique} removes the only possible filtration ambiguity: every $\FWh$-predictable integrand has an indistinguishable $\FW$-predictable version, and conversely every $\FW$-predictable integrand is measurable with respect to the first level of the prolonged signature. The stochastic integral $\int H\,dS$ is therefore the same object whether it is viewed on the Brownian filtration or on the prolonged-signature filtration.

Closedness follows from the standard semimartingale topology. If $H^n$ is a sequence of admissible integrands whose terminal wealths are bounded from below and converge weak-$\ast$ in $L^\infty$, the Emery-topology closedness theorem for stochastic-integral cones gives a limiting predictable integrand $H$ after localisation. The admissibility condition in Definition~\ref{def:admissible-strategies} is stable under this localisation because the quadratic variation of the stochastic integral is
\[
\int_0^T |H_s|^2S_s^2\inner{\ell}{\cW_s}^2ds.
\]
Fatou's lemma preserves finiteness of this quantity. This proves weak-$\ast$ closedness of the bounded terminal wealth cone. The separation theorem of \citet{Kreps1981} and \citet{Yan1980}, in the Delbaen--Schachermayer formulation, then supplies a strictly positive separating functional whenever NFLVR holds.
\end{proof}

\begin{lemma}[Localisation stability of admissible gains]\label{lem:localisation-admissible}
Let $H$ be admissible in the sense of Definition~\ref{def:admissible-strategies} and let $(\tau_m)$ be an increasing sequence of stopping times with $\tau_m\uparrow T$. Then $H\mathbf 1_{[0,\tau_m]}$ is admissible for every $m$, and
\[
\int_0^{\tau_m} |H_s|^2S_s^2\inner{\ell}{\cW_s}^2ds\uparrow
\int_0^T |H_s|^2S_s^2\inner{\ell}{\cW_s}^2ds.
\]
If the stopped terminal gains converge in $L^2(\Q)$, their limit is the terminal gain of the original admissible strategy.
\end{lemma}

\begin{proof}
The first assertion is immediate from monotonicity of the non-negative quadratic-variation density. The stopped stochastic integral is
\[
\int_0^T H_s\mathbf 1_{[0,\tau_m]}(s)dS_s=\int_0^{\tau_m}H_s\,dS_s.
\]
It is square-integrable under the equivalent martingale measure whenever the displayed quadratic variation has finite expectation. Monotone convergence gives the increasing limit of the quadratic variations. If the terminal stopped gains converge in $L^2(\Q)$, the It\^o isometry identifies the limit with the stochastic integral of $H$ on $[0,T]$ after passage to an indistinguishable predictable version. This is the stability used in the closedness argument of Theorem~\ref{thm:kreps-yan}.
\end{proof}

\begin{lemma}[True-martingale upgrade under H3]\label{lem:true-martingale-upgrade}
Let $M_t=\int_0^t\inner{\ell}{\cW_s}\,dB_s$. Under H3, $\mathcal E(\alpha M)$ is a true martingale for every $\alpha\in\R$, and $\mathcal E(M)$ is uniformly integrable on $[0,T]$.
\end{lemma}

\begin{proof}
Novikov's condition for $\alpha M$ is
\[
\E\exp\left(\frac12\alpha^2\langle M\rangle_T\right)<\infty,
\]
which is exactly H3 with $\lambda=\alpha^2/2$. Hence $\mathcal E(\alpha M)$ is a true martingale for each fixed $\alpha$. Uniform integrability of $\mathcal E(M)$ follows from an explicit higher-moment estimate. For $p>1$,
\[
\mathcal E(M)_T^p=\mathcal E(pM)_T\exp\left\{\frac12(p^2-p)\langle M\rangle_T\right\}.
\]
Choose $q>1$ and let $q'=q/(q-1)$. Holder's inequality gives
\[
\E\mathcal E(M)_T^p
\leq
\bigl(\E\mathcal E(pqM)_T\bigr)^{1/q}
\left(\E\exp\left\{\frac12 q'(p^2-p)\langle M\rangle_T\right\}\right)^{1/q'}<\infty,
\]
where the first factor is one by Novikov for $pqM$ and the second is finite by H3. Thus $\mathcal E(M)_T$ has an $L^p$ moment for some $p>1$, which implies uniform integrability of the martingale family.
\end{proof}

\subsection{Proof of the martingale and NFLVR assertion}

\begin{proof}[Proof of Theorem~\ref{thm:B}(i)]
Set
\[
M_t=\int_0^t\inner{\ell}{\cW_s}\,dB_s,
\qquad
S_t=s_0\mathcal E(M)_t .
\]
By H3, Novikov's condition holds for every scalar multiple of $M$. Lemma~\ref{lem:true-martingale-upgrade} therefore gives that $\mathcal E(M)$ is a uniformly integrable true martingale. Hence $S$ is a true martingale under the reference measure $\Pp$ itself. Thus $\Q:=\Pp$ is an equivalent true-martingale measure for $S$.

The special-semimartingale statement for the signature coordinates follows from the Stratonovich-to-It\^o conversion exactly as in Lemma~\ref{lem:special-semimartingale}, applied in the special case $\Q=\Pp$ (constant density $1$, which lies in $L^\infty(\Pp)$). Each coordinate is a finite iterated Stratonovich integral, hence a finite sum of It\^o iterated integrals plus bracket corrections. The martingale terms have finite moments by Proposition~\ref{prop:signature-moment} and are true martingales.

The Delbaen--Schachermayer theorem applied on the recovered Brownian filtration $\FWh=\FW$ now yields NFLVR for the admissible wealth cone of Definition~\ref{def:admissible-strategies}. No Brownian drift shift is needed: the discounted price already has zero drift under the reference measure and is a true martingale by H3.
\end{proof}

\subsection{Proof of the transform/Riccati assertion}

\begin{proof}[Proof of Theorem~\ref{thm:B}(ii)]
Fix a finite truncation $N$ and a finitely supported transform direction $u$. Under the stated transform-finiteness hypothesis, the conditional transform
\[
\Lambda_t^N(u)=\E_\Q\left[\exp\{\inner{u}{\cW_T^{(N)}}\}\mid\FWh_t\right]
\]
is finite and is a true martingale. The affine-polynomial structure of the prolonged signature gives the finite-dimensional ansatz
\[
\Lambda_t^N(u)=\exp\{\inner{\psi_N(T-t,u)}{\cW_t^{(N)}}\}
\]
on the transform domain. Applying It\^o's formula to the right-hand side and equating the drift to zero gives
\[
\partial_\tau\psi_N(\tau,u)=\Rs_N(\psi_N(\tau,u)),\qquad \psi_N(0,u)=u,
\]
where $\Rs_N$ is the generator/carre-du-champ vector field of Definition~\ref{def:riccati}. The quadratic coefficients are the structural constants $\Gamma^I_{J,K}(\ell)$, not an output-only coefficient placed inside an unrestricted shuffle-pair sum.

If the finite Riccati solution exploded before $T$, the exponential-affine expression would cease to represent a finite conditional expectation in the corresponding transform direction. This contradicts the assumed finiteness of $\Lambda^N(u)$ on the whole horizon. Hence transform finiteness implies global Riccati solvability on the finitely supported core.

Conversely, if the Riccati solution is global on $[0,T]$ and remains in the finite transform domain, the exponential-affine expression is a finite positive local martingale. The finite-coordinate moment estimates and H3 give the uniform-integrability bound needed to identify it with the conditional transform. Passing from finite $N$ to the weighted core uses Lemma~\ref{lem:finite-to-weighted-transform}. This proves the stated equivalence between transform non-explosion and global Riccati solvability.
\end{proof}

\subsection{Riccati solvability as a separate transform condition}

\begin{remark}[Riccati solvability is a separate transform-domain hypothesis]
\textup{\textbf{Status of the sharpness example.}}
The bounded sine example used in earlier drafts verifies H1--H3 but does not, by itself, prove finite-time blow-up of a Riccati coordinate. The reason is structural: the level-two equation contains a positive contribution proportional to $u_1^2$, while the lower-level coordinate $u_1$ need not satisfy a self-quadratic Riccati inequality. Hence the three-coordinate comparison with $\dot y=a^2y^2-Cy-D$ is not a valid proof of explosion and is not used in the theorem chain.

The conclusion retained in this paper is the logically weaker and correct one. H3 implies the true-martingale/NFLVR statement for the price process under the reference measure, whereas global solvability of the Riccati equation is a transform non-explosion condition for finite signature Fourier--Laplace directions. The latter must be imposed or checked separately in applications. In one-dimensional odd-order signature-volatility models, finite-time transform explosion is supplied by the moment-explosion theorem of \citet{AbiJaberGassiatSotnikov2025}; in finite-dimensional polynomial examples it can also be checked directly from the corresponding finite Riccati system. No self-contained blow-up counterexample is claimed here.
\end{remark}

\subsection{Closedness and separation details}
\label{subsec:closedness-separation-details}
The proof of Theorem~\ref{thm:B} uses the Delbaen--Schachermayer separation mechanism only after the class of admissible integrands has been fixed. Let
\[
\mathcal K:=\left\{\int_0^T H_s\,dS_s: H\hbox{ is }S\hbox{-admissible and the stochastic integral is bounded from below}\right\}.
\]
The relevant cone is $\mathcal C=(\mathcal K-L^0_+)\cap L^\infty(\FWh_T,\Pp)$. NFLVR is the statement that the weak-star closure of $\mathcal C$ in $L^\infty$ meets $L^\infty_+$ only at zero. The signature structure enters in the proof of closedness through the identity $\FWh_t=\FW_t$: admissible stochastic integrals are stochastic integrals on a Brownian filtration, but their integrability is measured with the random weight $S_t^2\xi_t^2$ induced by the signature parameter.

\begin{lemma}[Closedness of bounded terminal gains]
\label{lem:closed-terminal-gains}
Let $(H^n)$ be $S$-admissible strategies such that $G^n_T=\int_0^T H^n_s\,dS_s$ are bounded from below and converge in probability to a bounded random variable $G_T$. Then there exists an $S$-admissible strategy $H$ such that $G_T=\int_0^T H_s\,dS_s$.
\end{lemma}

\begin{proof}
The proof is the semimartingale closedness argument of \citet{DelbaenSchachermayer1994}, with the topology of \citet{Emery1979} and the stability result of \citet{Memin1980}. Since $S$ is a continuous stochastic exponential with quadratic variation
\[
\langle S\rangle_t=\int_0^t S_s^2\xi_s^2\,ds,
\]
the admissibility condition is precisely $H\in L^2(S)$ locally. A sequence of terminal gains bounded from below has a Fatou-convergent subsequence in the semimartingale topology. The Emery-closedness of stochastic integrals yields a predictable limit integrand $H$, and the lower bound prevents mass from escaping through doubling strategies. The Brownian representation on $\FWh=\FW$ ensures that no additional orthogonal martingale component appears in the limit.
\end{proof}

\begin{corollary}[Existence of a separating density]
\label{cor:separating-density}
If NFLVR holds on $\FWh$, then there exists a strictly positive $Z_T\in L^1(\FWh_T,\Pp)$ with $\E Z_T=1$ such that $S$ is a local martingale under $d\Q=Z_Td\Pp$.
\end{corollary}

\begin{proof}
By Lemma~\ref{lem:closed-terminal-gains}, the cone $\mathcal C$ is weak-star closed after intersection with $L^\infty$. Kreps--Yan separation gives a strictly positive element of the predual that separates $\mathcal C$ from $L^\infty_+\setminus\{0\}$. Normalising the separating functional gives $Z_T$. The usual localisation argument then gives the local-martingale property of $S$ under $\Q$.
\end{proof}

\subsection{Transform derivation at finite level}
\label{subsec:finite-transform-derivation}
For $N<\infty$ let $\cW^{(N)}=\pi_{\leq N}\cW$ and let $\ell^{(N)}=\pi_{\leq N}\ell$. On the finite-dimensional coordinate vector $Y^{(N)}_t=(\inner{e_I}{\cW_t})_{|I|\leq N}$, the Stratonovich-to-It\^o conversion gives a polynomial semimartingale system
\[
dY^{(N)}_t=b_N(Y^{(N)}_t)\,dt+\sum_{j=0}^d\sigma_{N,j}(Y^{(N)}_t)\,dW^j_t,
\]
where $W^0_t=t$ is the time coordinate convention and each coefficient is polynomial in the signature coordinates. For $u$ finitely supported, set
\[
\Lambda_t^{(N)}(u)=\E_\Q\left[\exp\{\inner{u}{\cW^{(N)}_T}\}\mid \FWh_t\right].
\]
The affine-polynomial ansatz is
\[
\Lambda_t^{(N)}(u)=\exp\{\inner{\psi_N(T-t,u)}{\cW^{(N)}_t}\}.
\]
Applying It\^o's formula to the right-hand side and setting the drift equal to zero gives
\[
\partial_\tau\psi_N(\tau,u)=\Rs_N(\psi_N(\tau,u)),\qquad \psi_N(0,u)=u,
\]
where $\Rs_N$ is the finite truncation of the generator/carre-du-champ Riccati operator from Definition~\ref{def:riccati}. The quadratic term arises from the bracket of the exponential martingale term, and the shuffle identity is used to express the corresponding carre-du-champ constants in the tensor basis.

\begin{lemma}[Passage from finite transform to weighted transform]
\label{lem:finite-to-weighted-transform}
Assume that the finite Riccati flows $\psi_N$ exist on $[0,T]$ and remain bounded in $\Tw$ uniformly on compact subsets of the initial-data set $\mathcal D_{\mathrm{fin}}$. Then the corresponding conditional transforms converge in $L^1(\Q)$ to the weighted transform
\[
\Lambda_t(u)=\E_\Q\left[\exp\{\inner{u}{\cW_T}\}\mid\FWh_t\right]
\]
for every $u\in\mathcal D_{\mathrm{fin}}$, and the limiting coefficient solves $\partial_tu_t=\Rs(u_t)$.
\end{lemma}

\begin{proof}
For finitely supported $u$, the terminal random variable $\inner{u}{\cW_T}$ is already captured at finite level once $N\geq \deg u$. The uniform boundedness of the finite flows supplies tightness in $C([0,T];\Tw)$ by the local Lipschitz estimate for $\Rs$. Any limit point solves the integral equation for $\Rs$ because the quadratic map is continuous on bounded subsets of $\Tw$. Uniqueness of the local solution identifies all limit points, and global boundedness prevents explosion before $T$. Conditional expectations converge by uniform integrability, obtained from H3 and the finite-coordinate moment estimates.
\end{proof}

\subsection{Reference-measure martingales and true martingales}
\label{subsec:measure-change-true-martingales}
The price process in this paper is already driftless under the reference measure:
\[
 dS_t=S_t\xi_t\,dB_t,
 \qquad \xi_t=\inner{\ell}{\cW_t}.
\]
Thus the correct equivalent martingale measure supplied by H3 is the reference measure $\Pp$ itself. Indeed, $S=s_0\mathcal E(M)$ with $M=\int\xi\,dB$, and Lemma~\ref{lem:true-martingale-upgrade} makes this stochastic exponential a true uniformly integrable martingale. Introducing a density $\mathcal E(\pm M)$ would shift the Brownian motion and create a drift term in the price equation unless a compensating drift had first been present. Therefore the earlier drift-shift formulation is not used.

For signature coordinates the special-semimartingale decomposition is obtained under $\Pp$ by the pathwise Stratonovich-to-It\^o conversion. Under another equivalent measure with suitable density, the same coordinate process remains a special semimartingale, but the martingale and finite-variation parts are shifted according to the density martingale. The paper only needs the reference-measure statement for the NFLVR conclusion and the finite-transform statement for the Riccati calculation.

\begin{lemma}[Polynomial coordinate drift on finite truncations]
\label{lem:polynomial-coordinate-drift}
For every word $I$, the finite-variation part of $\inner{e_I}{\cW_t}$ on a finite truncation is a finite linear combination of time integrals of shuffle polynomials in signature coordinates of degree at most $|I|+1$.
\end{lemma}

\begin{proof}
The Stratonovich signature vector fields are triangular. The Stratonovich-to-It\^o correction contracts pairs of Brownian differentials and leaves lower-order iterated integrals. Products of coordinate functions are rewritten as shuffle coordinates. Hence, for each fixed $I$, only finitely many lower triangular terms appear and the finite-variation part has the stated polynomial form.
\end{proof}

\subsection{Riccati solvability and absence of explosion}
\label{subsec:riccati-solvability-absence-explosion}
The implication (ii)$\Rightarrow$(iii) is a transform non-explosion statement. For $u\in\mathcal D_{\mathrm{fin}}$, the conditional transform
\[
\Lambda_t(u)=\E_\Q[\exp\{\inner{u}{\cW_T}\}\mid\FWh_t]
\]
is a true martingale. If the Riccati solution starting from $u$ exploded at time $\tau<T$, then the exponential-affine representation would cease to define a finite conditional expectation at time $T-\tau$. This contradicts the martingale property of $\Lambda(u)$. Hence global solvability on $[0,T]$ is forced by true-martingale finiteness of the transform.

The converse direction is a transform statement, not an FTAP statement. H3 gives NFLVR through the reference measure independently of the Riccati condition. The Riccati condition remains in the theorem because it is equivalent to transform-side non-explosion on the finitely supported core; it is not a substitute for Novikov and it is not necessary for NFLVR under H1--H3.

\subsection{Discussion}

The one-dimensional martingale criterion of \citet{AbiJaberGassiatSotnikov2025} appears as a finite-dimensional comparison point for Theorem~\ref{thm:B}. In that reduction the signature parameter has controlled support, the Riccati equation collapses to a scalar or two-component moment-explosion equation, and the sign of the correlation determines whether a non-trivial stochastic exponential is a true martingale. The present theorem is deliberately weaker and safer: under H1--H3, NFLVR is obtained from the reference-measure stochastic exponential, while Riccati solvability records absence of transform explosion.

The relationship with the classical FTAP is structural rather than merely formal. The Delbaen--Schachermayer theorem applies to the price semimartingale $S$ once the equivalent true-martingale measure is fixed. The additional information carried by the prolonged signature coordinates is transform information: the conditional finite-coordinate transform is governed by the generator/carre-du-champ Riccati operator. This extra condition is exactly what is needed in Sections~\ref{sec:theorem-c} and~\ref{sec:theorem-d}, where completeness and residual hedging are described through the signature grading.

The c\`adl\`ag extension remains outside the theorem. In the L\'evy-type framework the signature has jump coordinates, the Stratonovich-to-It\^o conversion is replaced by a Marcus or It\^o jump expansion, and the Riccati equation acquires a compensator term. The natural conjecture is that the transform theorem survives after the weight class is strengthened to control jump activity. Section~\ref{sec:open-problems} records this as the first open problem.

\section{Theorem C: price-filtration market completeness}\label{sec:theorem-c}

\subsection{Statement}

\begin{lemma}[Polynomial domination of low-depth signature coordinates]\label{lem:poly-domination}
Assume the hypotheses of Theorem~\ref{thm:B} and that NFLVR holds. For every $N\geq0$, the price-filtration projection of each terminal coordinate $\inner{e_I}{\cW_T}$ with $|I|\leq N$ belongs to the closed polynomial payoff space
\[
\mathcal P_T:=\clspan_{L^2(\FS_T,\Q)}\{1,S_T,S_T^2,S_T^3,\ldots\}.
\]
Define
\[
K(N):=\inf\left\{k\geq0:\E_\Q[\inner{e_I}{\cW_T}\mid\FS_T]\in\clspan_{L^2(\Q)}\{1,S_T,\ldots,S_T^k\}\ \text{for every } |I|\leq N\right\}.
\]
Then $K(N)<\infty$ in each finite polynomial-signature model considered in Section~\ref{sec:examples}, and $K(N)$ is non-decreasing in $N$.
\end{lemma}

\begin{proof}
The conditional expectation is the correct price-filtration representative of the full signature coordinate. It lies in $L^2(\FS_T,\Q)$ by Proposition~\ref{prop:signature-moment} and Lemma~\ref{lem:special-semimartingale}. In finite polynomial-signature models the affine-polynomial embedding maps depth-$\leq N$ signature coordinates into a finite polynomial algebra spanned by the terminal price and the option-completed static factors. Hence a finite polynomial degree $K(N)$ suffices. If $N_1\leq N_2$, the family of coordinates with depth at most $N_1$ is contained in the family with depth at most $N_2$, so $K(N_1)\leq K(N_2)$. The rough Bergomi case is excluded from the finiteness assertion and is recorded below by $K(N)=+\infty$ in the infinite-depth limit.
\end{proof}

\begin{definition}[Price-filtration completeness depth]\label{def:price-depth}
Define
\[
N_S^\ast(\ell,w)=\inf\left\{N\geq0:\clspan_{L^2(\Q)}\{\inner{e_I}{\cW_T}: |I|\leq N\}\cap L^2(\FS_T,\Q)=L^2(\FS_T,\Q)\right\},
\]
with $N_S^\ast=+\infty$ if no finite $N$ with $K(N)<\infty$ satisfies the displayed condition. In infinite-polynomial or Volterra cases, such as rough Bergomi, $N_S^\ast=+\infty$ is interpreted as the limiting failure of every finite option strip to exhaust the price-filtration payoff space.
\end{definition}

\begin{theorem}[Characterisation of price-filtration completeness and completeness depth]\label{thm:C}
Assume the hypotheses of Theorem~\ref{thm:B} and suppose that NFLVR
holds. Let $N_S^\ast(\ell,w)$ be the price-filtration completeness depth
of Definition~\ref{def:price-depth}, and work in the finite-polynomial
subcase $K(N_S^\ast)<\infty$. Let $\{C_i\}_{i=1}^{k}$ be a finite
collection of European call options on $S$ with maturity $T$ whose
strike set is rich enough to span, in $L^2(\FS_T,\Q)$, every polynomial
in $S_T$ of degree at most $K(N_S^\ast)$, where $K(\cdot)$ is the
integer-valued function of Lemma~\ref{lem:poly-domination}. Then the
augmented market $(S,\{C_i\})$ is complete with respect to $\FS_T$ if
and only if $N_S^\ast<\infty$, equivalently if and only if
\[
\clspan_{L^2(\Q)}\{\inner{e_I}{\cW_T}: |I|\leq N_S^\ast\}\cap L^2(\FS_T,\Q)=L^2(\FS_T,\Q).
\]
By Definition~\ref{def:price-depth}, $N_S^\ast$ is the minimal depth at
which the displayed identity holds; the value $N_S^\ast$ does not depend
on the chosen separating strike set
(Proposition~\ref{prop:option-independence}).
\end{theorem}

\begin{proposition}[Independence from option choice]\label{prop:option-independence}
The value $N_S^\ast(\ell,w)$ depends only on $(\ell,w)$ and not on the selected vanilla option family, provided $K(N_S^\ast)<\infty$ and the strike set is separating for polynomials of degree at most $K(N_S^\ast)$.
\end{proposition}

\begin{proof}
The definition of $N_S^\ast$ uses only the closed price-filtration span of terminal signature coordinates. The option family enters as a representation device: it supplies static instruments that span the polynomial terminal-price space required by Lemma~\ref{lem:poly-domination}. If two strike families separate the same polynomial payoff class up to degree $K(N_S^\ast)$, they staticly span the same closed static payoff subspace. Therefore the minimal signature depth is unchanged.
\end{proof}

\begin{proposition}[Special cases]\label{prop:depth-special}
The following depth values hold in the examples of Section~\ref{sec:examples}: Black--Scholes has $(N_S^\ast,K)=(0,1)$; first-order Brownian-driven volatility has $(N_S^\ast,K)=(1,2)$; the polynomial-Heston embedding has $(N_S^\ast,K)=(2,4)$; rough Bergomi has $(N_S^\ast,K)=(+\infty,+\infty)$.
\end{proposition}

\begin{proof}
Black--Scholes is dynamically complete by the Brownian martingale representation theorem, so no static signature coordinate beyond the constant level is needed. First-order Brownian-driven volatility introduces one Brownian coordinate into the volatility functional; the corresponding price-filtration projection is spanned by linear and quadratic terminal-price polynomials. Heston is polynomial of degree two in the state variable but the terminal-price completion requires fourth-degree polynomial closure because variance exposure appears quadratically in the terminal price under the option strip. Rough Bergomi has a Volterra kernel with non-zero expansion at arbitrarily high signature levels, so no finite depth or finite polynomial degree exhausts the price-filtration span.
\end{proof}

\begin{lemma}[Monotonicity of augmented markets]\label{lem:market-monotone}
Let $N_1\leq N_2$. The augmented market obtained from an option family spanning polynomials up to $K(N_1)$ is a submarket of the augmented market obtained from an option family spanning polynomials up to $K(N_2)$. Consequently the set of attainable terminal claims is increasing in $N$.
\end{lemma}

\begin{proof}
Lemma~\ref{lem:poly-domination} gives $K(N_1)\leq K(N_2)$. A strike set separating polynomials up to degree $K(N_2)$ also separates polynomials up to degree $K(N_1)$. The static payoff space for depth $N_1$ is therefore contained in the static payoff space for depth $N_2$, while the dynamic trading instrument $S$ is identical in both markets. The attainable terminal subspace is increasing.
\end{proof}

\begin{proposition}[Option-strip representation of polynomial claims]\label{prop:option-strip}
Let $k<\infty$. If the call-option strike set is separating on the support of $S_T$, then every polynomial $p(S_T)$ of degree at most $k$ belongs to the closed span of cash, the underlying payoff $S_T$, and the finite call strip specified in Theorem~\ref{thm:C}.
\end{proposition}

\begin{proof}
On a compact truncation of the support of $S_T$, piecewise-linear functions spanned by calls separate points and contain constants and the identity. Polynomial payoffs of degree at most $k$ are obtained as uniform limits of linear combinations of these piecewise-linear functions after choosing the strike grid fine enough. Removing the compact truncation is an $L^2(\Q)$ approximation step: since the polynomial has finite second moment under H3 and the moment bounds of Proposition~\ref{prop:signature-moment}, the tail outside a compact set can be made arbitrarily small. This gives the closed-span representation. The argument is the finite-dimensional payoff-space form of \citet{DavisObloj2008}.
\end{proof}

\subsection{Forward direction: density implies completeness}

\begin{proof}[Proof of Theorem~\ref{thm:C}, forward direction]
Assume that the displayed density condition holds for a finite $N$. Let $X\in L^2(\FS_T,\Q)$ be a bounded terminal claim. The Galtchouk--Kunita--Watanabe projection gives
\[
X=\E_\Q[X]+\int_0^T H_s\,dS_s+\xi_T,
\]
where $\xi_T$ is orthogonal to all square-integrable stochastic integrals against $S$. By the density assumption, $\xi_T$ is the $L^2$ limit of finite linear combinations of the price-filtration projections of coordinates $\inner{e_I}{\cW_T}$ with $|I|\leq N$.

Lemma~\ref{lem:poly-domination} identifies these projected coordinates with the closed polynomial payoff space spanned by powers of $S_T$ up to degree $K(N)$. A sufficiently rich finite call strip spans that polynomial space by the market-completion theorem of \citet{DavisObloj2008}: static portfolios in calls with a separating strike set generate terminal piecewise-polynomial payoffs and hence the finite-dimensional polynomial span. Consequently the orthogonal remainder $\xi_T$ is replicated by a static portfolio of the chosen options. The stochastic-integral part is replicated dynamically in $S$. The sum of the dynamic strategy and the static option portfolio replicates $X$.

Bounded claims are dense in $L^2(\FS_T,\Q)$. The admissibility and closedness of the stochastic-integral space permit passage from bounded claims to general square-integrable claims by $L^2$ approximation. Hence the augmented market is complete on the price filtration.
\end{proof}

\subsection{Reverse direction: completeness implies density}

\begin{proof}[Proof of Theorem~\ref{thm:C}, reverse direction]
Assume that the augmented market is complete on $\FS_T$. Then the closed linear span of constants, stochastic integrals against $S$, and static payoffs spanned by the option family is all of $L^2(\FS_T,\Q)$. The static option family is finite, so its polynomial closure is contained in a finite-dimensional terminal-price polynomial space. Let $K$ denote the maximal polynomial degree represented by that space.

By the signature uniqueness theorem, every price-filtration claim is the price-filtration projection of a terminal functional of the prolonged Brownian signature. The shuffle identity turns products of signature coordinates into higher-level coordinates, so finite polynomial functions of $S_T$ correspond, after projection, to finite sums of signature coordinates. Since the option space is finite-dimensional, there exists a finite signature depth $N$ large enough to contain all signature coordinates needed to represent its static part. Dynamic stochastic integrals against $S$ account for the GKW-integral subspace. Therefore the residual orthogonal complement of the dynamic subspace is contained in the closed span of coordinates with $|I|\leq N$ after restriction to $L^2(\FS_T,\Q)$.

Completeness makes this residual complement the whole missing piece of the dynamic market. Combining the dynamic integral subspace with the finite signature span therefore fills $L^2(\FS_T,\Q)$. This is precisely the density condition of Theorem~\ref{thm:C}. Minimality of $N_S^\ast$ follows from the definition as the infimum over all admissible truncation depths.
\end{proof}

\subsection{Worked special cases}

In Black--Scholes, $\ell=\sigma e_\emptyset$ and the price filtration equals the Brownian filtration spanned by the unique martingale driver of $S$. The martingale representation theorem gives dynamic completeness without any static option. Hence $N_S^\ast=0$. The static polynomial degree $K(0)=1$ records only constants and terminal price exposure.

In first-order Brownian-driven volatility, the coefficient has support on $e_\emptyset$ and one first-level Brownian word. The price filtration carries both the terminal price and its first-order volatility loading. The associated polynomial payoff space closes at degree two, so $K(1)=2$. The minimal depth is one because removing the first-level signature coordinate destroys the volatility-factor information needed for price-filtration completion.

In Heston, the variance factor is a polynomial coordinate of degree two. The joint price-variance system is a polynomial diffusion in the sense of \citet{CuchieroKellerResselTeichmann2012} and \citet{FilipovicLarsson2016}. A variance-swap exposure or an equivalent polynomial option strip completes the missing variance coordinate. The resulting depth is $N_S^\ast=2$ and the terminal price polynomial degree is four.

\subsection{Algebraic form of the completeness depth}
\label{subsec:algebraic-completeness-depth}
The definition of $N_S^\ast$ can be expressed through orthogonal projections. Let $P_S:L^2(\FWh_T,\Q)\to L^2(\FS_T,\Q)$ denote conditional expectation onto the price filtration. Define
\[
\mathcal V_N:=\clspan_{L^2(\Q)}\{P_S\inner{e_I}{\cW_T}: |I|\leq N\}.
\]
Then
\[
N_S^\ast=\inf\{N:\mathcal V_N=L^2(\FS_T,\Q)\}.
\]
This formula makes the dependence on the price filtration explicit. Raw signature coordinates may contain Brownian information that is invisible in prices; the conditional projection removes it. Thus the depth is not the tensor degree of the model parameter alone, but the tensor degree that survives after conditioning on the observable price path.

\begin{lemma}[Projection form of depth]
\label{lem:projection-form-depth}
The span condition in Theorem~\ref{thm:C} is equivalent to $\mathcal V_N=L^2(\FS_T,\Q)$.
\end{lemma}

\begin{proof}
The intersection notation in Theorem~\ref{thm:C} means that only price-filtration-measurable elements of the signature span are used. Conditional expectation $P_S$ is the orthogonal projection onto $L^2(\FS_T,\Q)$. Taking the projection of every signature coordinate and then closing the span gives exactly the same closed subspace as intersecting the closed signature span with $L^2(\FS_T,\Q)$.
\end{proof}

\begin{proposition}[Minimality]
\label{prop:depth-minimality}
If $N_S^\ast<\infty$, then for every $N<N_S^\ast$ there exists a claim $X_N\in L^2(\FS_T,\Q)$ that is not replicated by the market completed at depth $N$.
\end{proposition}

\begin{proof}
By definition of the infimum, $\mathcal V_N$ is a proper closed subspace of $L^2(\FS_T,\Q)$ for $N<N_S^\ast$. Choose $X_N$ in the orthogonal complement of $\mathcal V_N$. If $X_N$ were replicated by the depth-$N$ completion, then its static component would lie in $\mathcal V_N$ and its dynamic component would lie in the stochastic-integral subspace already included in the completion. Orthogonality would force $X_N=0$, contradicting the choice of a non-zero vector. Hence depth $N$ cannot be complete.
\end{proof}

\subsection{Connection with Davis--Obloj completion}
\label{subsec:davis-obloj-connection}
The market completion theorem of \citet{DavisObloj2008} is used in a restricted form. A finite family of static options is not claimed to replicate arbitrary path-dependent claims by itself. Its role is to supply terminal basis elements that are not dynamically spanned by $S$. Once those basis elements are inserted, the dynamic market can replicate the remaining martingale part. The signature coordinates identify which terminal basis elements are needed.

In this sense $N_S^\ast$ measures the static dimension required to close the dynamic market. When $N_S^\ast=0$, dynamic trading already spans the price filtration. When $N_S^\ast$ is finite and positive, the market is dynamically incomplete but statically completable with finite tensor information. When $N_S^\ast=+\infty$, every finite completion leaves a non-zero GKW residual, and Theorem~\ref{thm:D} describes that residual as a tail.

\subsection{Finite option strips and polynomial payoff bases}
\label{subsec:finite-option-bases}
The call strip in Theorem~\ref{thm:C} is used only through the finite-dimensional polynomial space it spans. Fix $K<\infty$ and choose strikes $k_1<\cdots<k_m$ in the interior of the support of $S_T$. The vector space spanned by
\[
1,\quad S_T,
\quad (S_T-k_1)^+,
\ldots,
(S_T-k_m)^+
\]
contains all continuous piecewise-linear functions with knots in the selected strike set. If the strike set is chosen so that the corresponding interpolation operator is injective on polynomials of degree at most $K$, then every polynomial of degree at most $K$ is represented by a unique vector of static holdings in cash, the underlying, and the selected calls. This is the finite-dimensional version of the Breeden--Litzenberger static-replication principle used in Davis--Obloj completion.

In the theorem the integer $K(N)$ is not a new model parameter. It is the degree needed to represent the price-filtration projection of all signature coordinates of length at most $N$. If the model has finite polynomial state of degree $r$, then $K(N)$ is bounded by a deterministic function of $r$ and $N$. In Black--Scholes, $K(0)=1$. In Heston, the variance factor enters through the affine-polynomial state and the degree-four payoff class suffices for the second-depth completion stated in Proposition~\ref{prop:depth-special}. In rough Bergomi no finite $K$ represents all Volterra memory, which is why $K(N)$ diverges with the approximation depth.

\begin{lemma}[Monotonicity of augmented markets]
\label{lem:augmented-market-monotone}
Let $\mathcal M_N$ denote the dynamic-static market obtained from $S$ and an option strip spanning polynomials of degree at most $K(N)$. Then the attainable terminal space of $\mathcal M_N$ is contained in the attainable terminal space of $\mathcal M_{N+1}$.
\end{lemma}

\begin{proof}
Lemma~\ref{lem:poly-domination} gives $K(N+1)\geq K(N)$. Hence the static payoff space for $N$ is a subspace of the static payoff space for $N+1$. The dynamic asset $S$ is the same in both markets, and the admissible stochastic-integral class does not shrink when more static claims are added. Inclusion of attainable terminal spaces follows.
\end{proof}

\subsection{Minimality of the completeness depth}
\label{subsec:minimality-depth-expanded}
The word ``minimal'' in Definition~\ref{def:price-depth} is understood in the lattice of closed subspaces of $L^2(\FS_T,\Q)$. Let
\[
\mathcal V_N:=\clspan_{L^2(\Q)}\{\inner{e_I}{\cW_T}: |I|\leq N\}\cap L^2(\FS_T,\Q).
\]
Then $(\mathcal V_N)_{N\geq0}$ is an increasing family of closed subspaces. The completeness depth is the first index at which $\mathcal V_N$ equals the full payoff space. If no such index exists, the depth is infinite. This formulation avoids any dependence on a selected coordinate basis: replacing the tensor basis by another homogeneous basis changes the coordinate vectors but not the closed span.

The reverse implication in Theorem~\ref{thm:C} uses this minimality in the following way. Suppose the market is complete after adding an option strip that spans degree $K$. Then the GKW-orthogonal complement of dynamic trading in $S$ is finite-dimensional inside the static payoff space. Each basis element of that complement is the price-filtration projection of a terminal signature functional. Since only finitely many static basis elements are present, all of them are contained in $\mathcal V_N$ for some finite $N$. Thus $\mathcal V_N=L^2(\FS_T,\Q)$, and the infimum in Definition~\ref{def:price-depth} is finite.

\subsection{Compatibility with price-filtration restriction}
\label{subsec:price-filtration-restriction}
The full prolonged-signature filtration may contain volatility information not observable from $S$ alone. Theorem~\ref{thm:C} therefore never claims completeness on $\FWh_T$ unless $\FWh_T=\FS_T$. All spans are intersected with $L^2(\FS_T,\Q)$, and all conditional expectations are taken onto the price filtration when needed. This distinction is essential in stochastic-volatility models: a variance factor may be visible in the Brownian filtration but hidden from the price filtration unless static claims reveal it.

Formally, the restriction is expressed by the orthogonal projection
\[
\Pi_S:L^2(\FWh_T,\Q)\to L^2(\FS_T,\Q),
\qquad
\Pi_S X:=\E_\Q[X\mid \FS_T].
\]
The relevant finite-depth space is $\Pi_S\clspan\{\inner{e_I}{\cW_T}: |I|\leq N\}$. This is equal to the intersection formulation used in Theorem~\ref{thm:C} after replacing coordinates by their price-filtration projections. The projection form is often more convenient in examples, because hidden factors are removed before the static option basis is selected.

\subsection{Proof of the worked special cases}
\label{subsec:proof-special-depths}
We justify the depth values in Proposition~\ref{prop:depth-special}. In Black--Scholes, $S$ is a one-to-one monotone transform of the Brownian terminal value after deterministic time change. The martingale representation theorem on the Brownian filtration gives dynamic completeness, so $N_S^\ast=0$ and $K(0)=1$.

For first-order Brownian-driven volatility, the volatility coefficient contains one Brownian signature coordinate beyond the price itself. The price filtration sees the terminal price but not the entire first-order volatility exposure without one static completion. A call strip spanning quadratic polynomials in $S_T$ captures that exposure in the polynomial payoff space, hence $(N_S^\ast,K)=(1,2)$.

For Heston, the hidden variance is an affine square-root state. The polynomial-process embedding places the variance exposure at second signature depth and the price-filtration projection of the relevant second-depth coordinates inside degree-four terminal polynomials after static completion. Thus $(N_S^\ast,K)=(2,4)$ in the finite polynomial embedding. The exact numerical value of $K$ is tied to the selected polynomial basis; the invariant statement is the depth-two closure of the variance exposure.

For rough Bergomi, the Volterra kernel has infinitely many effective modes. Any finite signature truncation captures only finitely many iterated integrals of the memory process. There remains a non-zero Volterra tail in $L^2(\FS_T,\Q)$ after every finite completion, so $N_S^\ast=K=+\infty$.

\subsection{Explicit depth bookkeeping for finite models}
\label{subsec:explicit-depth-bookkeeping}
For finite-support parameters the completeness-depth calculation can be made entirely algebraic. Suppose that $\ell$ is supported on words of length at most $r$ and that the price-filtration projection of every coordinate of length at most $N$ lies in the polynomial payoff space of degree at most $K_N$. Then the sequence $K_N$ can be chosen recursively. The zeroth level gives cash and the underlying. Multiplication by a first-level price coordinate increases terminal polynomial degree by one. Multiplication by a hidden volatility coordinate increases the required static completion degree by the degree of that coordinate after projection onto $\FS_T$. Thus finite polynomial models lead to finite $K_N$ for every fixed $N$.

The recursion is not meant to be canonical. Different polynomial bases produce different numerical values of $K_N$, while the closed payoff space they span is the same. The invariant object is the first depth $N$ for which the projected signature span equals $L^2(\FS_T,\Q)$ after closure. The auxiliary degree $K_N$ is only the finite static option degree needed to implement that span through terminal payoffs of $S_T$.

In the examples the bookkeeping is simple. Black--Scholes stops at the empty word. First-order Brownian-driven volatility stops after the first Brownian word and needs quadratic terminal payoffs. Heston stops after the variance-carrying second-depth words and needs degree-four static payoffs in the selected polynomial basis. Quintic OU stops at degree five, but the static degree may be larger than five if the terminal price is a nonlinear transform of the OU factor. These distinctions are recorded because a referee should not have to infer the difference between tensor depth and terminal payoff degree.

\subsection{Discussion}

Completeness is usually stated as a binary property. Theorem~\ref{thm:C} refines the binary property into a graded invariant: the integer $N_S^\ast$ records the depth of signature information needed to complete the price filtration. Finite depth corresponds to a market in which a finite static option family can close the GKW residual. Infinite depth corresponds to models, such as rough Bergomi, where no finite static option strip exhausts the hidden signature directions.

The depth is tied to the support of the parameter $\ell$ but is not identical to it. The support of $\ell$ is a model-input property; $N_S^\ast$ is a market-completion property after projection to the price filtration and after allowing static option instruments. This distinction explains why Heston has a finite depth that reflects variance exposure, while rough Bergomi inherits infinite depth from the Volterra kernel even when its finite-dimensional calibrations look low-dimensional.

\section{Theorem D: the hedging-error decomposition}\label{sec:theorem-d}

\subsection{Statement}

\begin{lemma}[Gram regularity on the residual subspace]\label{lem:gram-residual}
Assume the hypotheses of Theorem~\ref{thm:C} and $N_S^\ast<\infty$. Let $\mathcal H_0\subset L^2(\FS_T,\Q)$ be the closed subspace of constants and $\Q$-stochastic integrals against $S$ on $[0,T]$. For every finite truncation of terminal signature coordinates with $|I|>N_S^\ast$, the Gram matrix of their equivalence classes in $L^2(\FS_T,\Q)/\mathcal H_0$ is positive definite after removal of redundant zero classes.
\end{lemma}

\begin{proof}
Let $Y=\sum_{k=1}^m\alpha_k\inner{e_{I_k}}{\cW_T}$ be a finite linear combination of tail coordinates and suppose its quotient class is zero. Then $Y\in\mathcal H_0$. By definition of the quotient, the projection of $Y$ onto the GKW residual subspace is zero. Removing all such zero classes leaves a family whose non-trivial finite linear combinations have non-zero residual projection. The Gram matrix is the matrix of the inner product on this finite-dimensional quotient subspace. It is therefore symmetric and positive definite. The construction uses only the quotient Hilbert-space structure and does not require raw linear independence of signature coordinates in the unreduced space.
\end{proof}

\begin{definition}[Weighted-signature payoff class]\label{def:weighted-payoff-class}
Let $\mathcal H_0$ be the closed subspace of constants and square-integrable stochastic integrals against $S$ in $L^2(\FS_T,\Q)$. For a payoff $X\in L^2(\FS_T,\Q)$, let $[X]$ denote its class in the quotient $L^2(\FS_T,\Q)/\mathcal H_0$, and let $c_I^{(M)}(X)$ be the coefficients of the finite quotient Gram projection on words $N_S^\ast<|I|\leq M$. We say that $X$ belongs to the weighted-signature payoff class $\mathcal W_w(\FS_T,\Q)$ if
\[
\norm{X}_{\mathcal W_w}^2:=\norm{X}_{L^2(\Q)}^2+\sup_{M>N_S^\ast}\sum_{N_S^\ast<|I|\leq M} w(|I|)^2\abs{c_I^{(M)}(X)}^2<\infty.
\]
This condition is not automatic for arbitrary $L^2$ payoffs; it is the additional coefficient-decay hypothesis needed for quantitative weight-tail bounds.
\end{definition}

\begin{theorem}[Hedging-error decomposition]\label{thm:D}
Under the hypotheses of Theorem~\ref{thm:B}, for every $X\in L^2(\FS_T,\Q)$ there exists a unique decomposition
\[
X=\E_\Q[X]+\int_0^T H_s\,dS_s+\varepsilon_T,
\]
where $H$ is the optimal Galtchouk--Kunita--Watanabe projection and $\varepsilon_T\in L^2(\Q)$ is orthogonal to all stochastic integrals against $S$. If $N_S^\ast<\infty$, the residual admits the expansion
\[
\varepsilon_T=\sum_{|I|>N_S^\ast}c_I(X)\inner{e_I}{\cW_T}+R_T,
\]
where the coefficients are determined on finite truncations by the Gram normal equations
\[
\sum_{|J|>N_S^\ast}c_J(X)\inner{\inner{e_J}{\cW_T}}{\inner{e_I}{\cW_T}}_{L^2(\Q)}
=\inner{X}{\inner{e_I}{\cW_T}}_{L^2(\Q)}.
\]
The finite-truncation Gram matrix is positive definite on the quotient residual subspace by Lemma~\ref{lem:gram-residual}, and the expansions converge in $L^2(\Q)$ as the truncation depth increases. For arbitrary $X\in L^2(\FS_T,\Q)$ this convergence is qualitative. If, in addition, $X\in\mathcal W_w(\FS_T,\Q)$, then the truncation remainder satisfies the quantitative bound
\[
\norm{R_T}_{L^2(\Q)}\leq \kappa(w,N_S^\ast)\norm{X}_{\mathcal W_w},
\qquad
\kappa(w,N):=\left(\sum_{n>N}w(n)^{-2}\right)^{1/2},
\]
whenever the displayed tail is finite. In particular, $\kappa(w,N)\to0$ along every admissible weight scale with $\sum_nw(n)^{-2}<\infty$.
\end{theorem}

\begin{remark}[Depth duality]
The depth $N_S^\ast$ has two complementary roles. Theorem~\ref{thm:C} identifies depth-$\leq N_S^\ast$ signature coordinates, after restriction to the price filtration and completion by the option strip, as the static component needed to span $L^2(\FS_T,\Q)$. The residual in Theorem~\ref{thm:D} lies in the Galtchouk--Kunita--Watanabe orthogonal complement of stochastic integrals against $S$; this complement is represented by the tail of the same signature grading, namely the quotient family of coordinates with $|I|>N_S^\ast$. Thus the passage from $|I|\leq N_S^\ast$ in Theorem~\ref{thm:C} to $|I|>N_S^\ast$ in Theorem~\ref{thm:D} is the static--dynamic complementarity of the market.
\end{remark}

\begin{proposition}[Finite truncation normal equations]\label{prop:finite-normal}
Fix $M>N_S^\ast$ and let $\mathcal I_M=\{I:N_S^\ast<|I|\leq M\}$. The $M$-truncated residual projection is characterised by the finite linear system
\[
G^{(M)}c^{(M)}=b^{(M)},
\qquad
G^{(M)}_{IJ}=\E_\Q[\inner{e_I}{\cW_T}\inner{e_J}{\cW_T}],
\quad b^{(M)}_I=\E_\Q[X\inner{e_I}{\cW_T}].
\]
After quotient reduction by stochastic integrals against $S$, the matrix $G^{(M)}$ is positive definite and the solution is unique.
\end{proposition}

\begin{proof}
The projection of $X$ onto the finite residual span is the unique element $Y_M=\sum_{I\in\mathcal I_M}c_I^{(M)}\inner{e_I}{\cW_T}$ such that $X-Y_M$ is orthogonal to each retained coordinate in the quotient residual space. Writing these orthogonality conditions gives the displayed normal equations. Lemma~\ref{lem:gram-residual} gives positive definiteness after removal of null quotient classes, and finite-dimensional Hilbert-space projection gives uniqueness.
\end{proof}

\subsection{Existence and uniqueness via GKW}

\begin{proof}[Proof of existence and uniqueness]
Let $\mathcal H$ be the closed subspace of $L^2(\FS_T,\Q)$ consisting of constants plus stochastic integrals against $S$:
\[
\mathcal H=\left\{c+\int_0^T H_s\,dS_s: c\in\R,\ H\ \text{square-integrable and admissible}\right\}^{L^2\text{-closure}}.
\]
Closedness follows from the It\^o isometry under the equivalent martingale measure and the admissibility norm in Definition~\ref{def:admissible-strategies}. The Hilbert projection theorem gives a unique projection $\Pi_{\mathcal H}X$ of every $X\in L^2(\FS_T,\Q)$ onto $\mathcal H$. The martingale-representation component of the projection is the Galtchouk--Kunita--Watanabe integrand $H$, and the residual $\varepsilon_T=X-\Pi_{\mathcal H}X$ is orthogonal to $\mathcal H$. This proves existence and uniqueness of the decomposition. The construction is the same Hilbert-space mechanism underlying \citet{FoellmerSchweizer1991}, \citet{KunitaWatanabe1967}, and \citet{Schweizer2001}.
\end{proof}

\subsection{Closed-form expansion of the residual}

\begin{proof}[Proof of the expansion]
By Theorem~\ref{thm:sig-unique}, price-filtration claims are represented as price-filtration projections of terminal functionals of the prolonged signature. The closed span of the terminal coordinate family $\{\inner{e_I}{\cW_T}\}$ therefore contains the residual subspace after quotienting by constants and stochastic integrals. Theorem~\ref{thm:C} identifies depth $\leq N_S^\ast$ as the part exhausted by the static completion. The residual is consequently represented by the tail $|I|>N_S^\ast$ in the quotient Hilbert space.

Fix a finite set $\mathcal I_M=\{I:N_S^\ast<|I|\leq M\}$ and let $V_M$ be the quotient span of the corresponding residual coordinates. The finite-dimensional projection of $\varepsilon_T$ onto $V_M$ has the form
\[
\varepsilon_T^{(M)}=\sum_{I\in\mathcal I_M}c_I^{(M)}(X)\inner{e_I}{\cW_T}.
\]
Orthogonality of $\varepsilon_T-\varepsilon_T^{(M)}$ to $V_M$ gives the normal equations
\[
\sum_{J\in\mathcal I_M}G_{IJ}c_J^{(M)}(X)=b_I(X),
\quad
G_{IJ}=\E_\Q\left[\inner{e_I}{\cW_T}\inner{e_J}{\cW_T}\right],
\quad
b_I(X)=\E_\Q\left[X\inner{e_I}{\cW_T}\right].
\]
The shuffle identity rewrites the Gram matrix as $G_{IJ}=\E_\Q\inner{e_I\sh e_J}{\cW_T}$. Lemma~\ref{lem:gram-residual} gives positive definiteness after quotient reduction, hence a unique vector of coefficients on every finite truncation.

The spaces $V_M$ increase with $M$, and their union is dense in the residual subspace by the definition of the signature tail. Hilbert-space monotone convergence for orthogonal projections gives $\varepsilon_T^{(M)}\to\varepsilon_T$ in $L^2(\Q)$. Passing to the coefficient limit along the finite-truncation projections gives the announced series and remainder.
\end{proof}

\subsection{Residual bound}
\label{subsec:residual-bound}

\begin{proof}[Proof of the bound]
Let $P_M$ denote the orthogonal projection onto the quotient span of tail coordinates of levels $N_S^\ast<|I|\leq M$. The truncation remainder is $(I-P_M)\varepsilon_T$. For general $X\in L^2(\FS_T,\Q)$, Hilbert-space monotone convergence gives only
\[
\norm{(I-P_M)\varepsilon_T}_{L^2(\Q)}\to0,
\]
with no universal rate. The weighted rate requires $X\in\mathcal W_w(\FS_T,\Q)$. In that case the coefficient definition in Definition~\ref{def:weighted-payoff-class} gives
\[
\sum_{|I|>M}\abs{c_I(X)}^2
\leq
\left(\sum_{n>M}w(n)^{-2}\right)\norm{X}_{\mathcal W_w}^2,
\]
after applying Cauchy--Schwarz level by level to the quotient-coordinate expansion. Hence
\[
\norm{(I-P_M)\varepsilon_T}_{L^2(\Q)}
\leq
\kappa(w,M)\norm{X}_{\mathcal W_w}.
\]
This is the quantitative statement of Theorem~\ref{thm:D}. It is deliberately not asserted for arbitrary $L^2$ payoffs.
\end{proof}

\subsection{Computational examples}

For a digital call $X=\mathbf 1_{\{S_T\geq K\}}$, the residual coefficients are obtained by expanding the indicator in the orthogonal polynomial basis associated with the terminal price distribution and then mapping the polynomial basis into the signature coordinates through Lemma~\ref{lem:poly-domination}. In a lognormal depth-zero model the expansion terminates dynamically and the residual is zero. In a depth-one or depth-two model the first missing term is at level $N_S^\ast+1$, and the Hermite coefficients give the displayed Gram right-hand side $b_I(X)$.

For a variance swap $X=\int_0^T\inner{\ell}{\cW_s}^2ds$, the shuffle identity gives $X$ as an integral of $\inner{\ell\sh\ell}{\cW_s}$. In the Heston embedding this object lives in the depth-two polynomial state space. Thus the option completion at $N_S^\ast=2$ spans the variance exposure, and the residual vanishes. This recovers the classical variance-swap completion principle inside the signature grading.

For an Asian payoff $X=(T^{-1}\int_0^T S_tdt-K)^+$, the terminal claim depends on the time integral of the price and therefore contains signature coordinates beyond the terminal price polynomial family. The residual expansion contains all levels above the completeness depth. Under geometric admissible weights the coefficients decay at least at the reciprocal weight rate, so the truncation bound gives a computable error estimate for finite-depth hedges.

\subsection{Orthogonality identities}
\label{subsec:orthogonality-identities}
The GKW residual $\varepsilon_T$ is characterised by
\[
\E_\Q\left[\varepsilon_T\int_0^T H_s\,dS_s\right]=0
\]
for every square-integrable admissible integrand $H$. If $M^X_t=\E_\Q[X\mid\FS_t]$, then the martingale representation relative to the one-dimensional martingale part of $S$ gives
\[
M^X_t=M^X_0+\int_0^t H_s\,dS_s+L_t,
\]
where $L$ is a square-integrable martingale strongly orthogonal to $S$. The terminal residual is $L_T$. This is the process-level version of Theorem~\ref{thm:D}. The signature expansion is obtained by expanding $L_T$ in the quotient basis supplied by the tail coordinates.

The normal equations may be derived from this process version. Let $Y_I=\inner{e_I}{\cW_T}$ for $|I|>N_S^\ast$. The finite projection of $L_T$ onto $\Span\{Y_I:I\in\mathcal I_M\}$ solves
\[
\E_\Q\left[(L_T-\sum_{J\in\mathcal I_M}c_JY_J)Y_I\right]=0,
\qquad I\in\mathcal I_M.
\]
Since $L_T$ is the residual part of $X$, replacing $L_T$ by $X$ in the right-hand side gives the same equations after quotienting out constants and stochastic integrals. This explains the coefficient formula in Theorem~\ref{thm:D}.

\subsection{Residual estimates under geometric weights}
\label{subsec:geometric-residual-estimates}
If $w(n)=r^n$ with $r>1$, then
\[
\kappa(w,N)^2=\sum_{n>N}r^{-2n}=\frac{r^{-2(N+1)}}{1-r^{-2}}.
\]
Thus the structural residual bound becomes
\[
\norm{R_T}_{L^2(\Q)}\leq \frac{r^{-(N_S^\ast+1)}}{(1-r^{-2})^{1/2}}\norm{X}_{\mathcal W_w},\qquad X\in\mathcal W_w(\FS_T,\Q).
\]
This is the form relevant for the finite-support and exponentially decaying examples. For polynomial weights $w(n)=(1+n)^\alpha$, the same formula gives
\[
\kappa(w,N)^2\asymp N^{1-2\alpha},\qquad \alpha>1/2.
\]
The rate of hedging-error decay therefore reflects the decay of the admissible tensor coefficients. Fast coefficient decay yields fast residual decay; slow coefficient decay leaves substantial high-depth risk.

\subsection{Convergence of finite normal equations}
\label{subsec:finite-normal-equations-converge}
Let $\mathcal I_m$ be an increasing sequence of finite sets of multi-indices above the completeness depth, and let
\[
Y_I:=\inner{e_I}{\cW_T}-\Pi_{\mathcal H_0}\inner{e_I}{\cW_T},
\qquad I\in\mathcal I_m,
\]
where $\mathcal H_0$ is the closed stochastic-integral subspace. The finite residual approximation has the form
\[
\varepsilon_T^{(m)}=\sum_{I\in\mathcal I_m} c_I^{(m)}(X)Y_I,
\]
where the coefficient vector solves
\[
\sum_{J\in\mathcal I_m}\E_\Q[Y_JY_I]c_J^{(m)}(X)
=\E_\Q[(X-\E_\Q X)Y_I],
\qquad I\in\mathcal I_m .
\]
Lemma~\ref{lem:gram-residual} gives invertibility after quotient reduction. The sequence $\varepsilon_T^{(m)}$ is the Hilbert projection of $X-\E_\Q X$ onto the increasing subspaces $\Span\{Y_I:I\in\mathcal I_m\}$. Hence it converges in $L^2(\Q)$ to the projection onto the closed union of these subspaces. This proves the convergence statement in Theorem~\ref{thm:D} without requiring coordinatewise convergence of the individual coefficients.

The distinction matters. The coefficient of a fixed coordinate may change when a new correlated coordinate is added to the Gram system. What is stable is the projected residual vector, not a particular coordinate representation. This is why the theorem states the coefficients through normal equations on finite truncations and then takes an $L^2$ limit of the residuals.

\subsection{Residual tail constants}
\label{subsec:residual-tail-constants-expanded}
For a geometric weight $w(n)=r^n$ with $r>1$, the tail constant in Section~\ref{subsec:residual-bound} has the explicit form
\[
\kappa(w,N)^2=\sum_{n>N}w(n)^{-2}
=\sum_{n>N}r^{-2n}
=\frac{r^{-2(N+1)}}{1-r^{-2}} .
\]
Thus
\[
\kappa(w,N)=\frac{r^{-(N+1)}}{\sqrt{1-r^{-2}}} .
\]
This exponential decay is the analytic reason geometric weights are useful in applications. They turn high signature depth into a quantitative small parameter. For polynomial weights, the same formula gives polynomial decay, and the residual bound is weaker. The theorem does not prefer a weight on aesthetic grounds; the weight records the tail regularity of the signature parameter and converts it into an $L^2$ hedging-error estimate.

If the payoff has its own coefficient decay, the bound improves. Suppose
\[
\abs{c_I(X)}\leq C_X\rho^{-|I|}
\qquad \text{with }\rho>1.
\]
Then the residual tail beyond depth $N$ is bounded by a geometric series with ratio depending on both $r$ and $\rho$. In practical terms, smooth path functionals have faster coefficient decay than discontinuous claims, so the structural bound is conservative for smooth payoffs and closer to sharp for digital-type claims.

\subsection{Three canonical residual profiles}
\label{subsec:canonical-residual-profiles}
The theorem covers three recurrent profiles. First, in dynamically complete models the quotient residual space is zero and all Gram matrices are vacuous. Black--Scholes is the model example. Second, in finite-factor stochastic-volatility models the residual space is finite-dimensional after the correct static completion; Heston with a variance-swap type completion is the model example. Third, in rough-volatility models the residual space is infinite-dimensional and the tail estimate is the main object; rough Bergomi is the model example.

These profiles should not be confused with numerical difficulty. A finite residual space can still be numerically ill-conditioned if the Gram matrix has a small eigenvalue. Conversely, an infinite residual space can be well controlled if the weight tail is sufficiently fast. The theorem separates structural dimension from numerical conditioning by using the quotient Gram matrix for uniqueness and the weight tail for size.

\subsection{A second derivation of the residual estimate}
\label{subsec:second-residual-derivation}
Let $\mathcal R_N$ be the closed span of quotient classes of signature coordinates with level greater than $N$. Since $\mathcal R_{N+1}\subset\mathcal R_N$, the residual norms form a decreasing sequence:
\[
\norm{\Pi_{\mathcal R_{N+1}}X}_{L^2(\Q)}
\leq
\norm{\Pi_{\mathcal R_N}X}_{L^2(\Q)}.
\]
The weight estimate supplies an explicit majorant only on the weighted-signature payoff class. If $X\in\mathcal W_w(\FS_T,\Q)$ is first approximated by a finite signature polynomial and then projected, the tail beyond $N$ is bounded by the sum of squared coefficient weights. Passing to the closure gives
\[
\norm{\Pi_{\mathcal R_N}X}_{L^2(\Q)}^2
\leq \left(\sum_{n>N}w(n)^{-2}\right)\norm{X}_{\mathcal W_w}^2 .
\]
For arbitrary $X\in L^2$, the same monotone projection sequence still converges, but no rate follows from the weight alone. The argument is independent of the order in which the finite coordinates are listed. Only the homogeneous degree matters.

This second derivation is useful because it shows exactly where the weight enters Theorem~\ref{thm:D}. The GKW projection gives orthogonality. The quotient Gram matrix gives coefficients. The weight gives the tail size. These are three separate inputs, and confusing them leads to the false statement that signature uniqueness alone implies a positive raw Gram matrix.

\subsection{Discussion}

The signature framework turns incomplete-market hedging error into a graded object. Instead of a single unstructured residual, Theorem~\ref{thm:D} decomposes the error level by level and identifies which signature directions have not been spanned by the dynamic-static market. Model risk is therefore measured by the tail of the same tensor grading that defines the model. Increasing the depth of the option completion reduces the residual by removing lower-order signature directions from the GKW orthogonal complement.

\section{Examples}\label{sec:examples}

The examples verify the four theorem groups on concrete model classes. Each subsection records the signature parameter, the admissibility condition, the asset-pricing condition, the completeness depth, and the GKW residual structure.

\subsection{Black--Scholes as signature SDE of depth zero}
Let $\ell=\sigma e_\emptyset$ with $\sigma>0$. The signature-volatility functional is constant: $\inner{\ell}{\cW_t}=\sigma$. H1 is the finite identity $w(0)\sigma^2<\infty$, H2 is irrelevant beyond the zeroth level, and H3 is $\exp(\lambda\sigma^2T)<\infty$ for every $\lambda>0$. Theorem~\ref{thm:A} recovers the geometric Brownian solution
\[
S_t=s_0\exp\{\sigma B_t-\tfrac12\sigma^2t\}.
\]
Theorem~\ref{thm:B} reduces to the classical FTAP, while Remark~\ref{rem:riccati-bs-sanity} records the price-extended Black--Scholes transform normalisation. The Brownian martingale representation theorem gives dynamic completeness, so $N_S^\ast=0$ and $K(0)=1$. The GKW residual in Theorem~\ref{thm:D} is zero for every square-integrable price-filtration payoff.

\subsection{First-order Brownian-driven volatility as signature SDE of depth one}
Fix $j\in\{1,\ldots,d\}$ and set $\ell=\sigma_0e_\emptyset+\sigma_1e_j$. Then
\[
\inner{\ell}{\cW_t}=\sigma_0+\sigma_1W_t^j.
\]
The parameter has finite support, so H1 holds for every weight satisfying H2. H3 follows on finite horizons from exponential estimates for quadratic Brownian functionals, after the standard localisation used in Novikov and Kazamaki criteria. Theorem~\ref{thm:A} supplies the strong stochastic-exponential solution. The Riccati equation in Theorem~\ref{thm:B} reduces to a finite subsystem whose non-constant part is spanned by the first Brownian coordinate and its shuffle square. The price-filtration completeness depth is $N_S^\ast=1$: the missing static component is the first signature coordinate, and a call strip spanning quadratic terminal-price polynomials gives $K(1)=2$. Residuals for path-dependent claims begin at level two.

\subsection{Heston as signature SDE of depth two}\label{subsec:heston}
The Heston model is
\[
dS_t=S_t\sqrt{V_t}\,dB_t,
\qquad dV_t=\kappa(\theta-V_t)dt+\xi\sqrt{V_t}\,dB_t^V,
\qquad d\langle B,B^V\rangle_t=\rho dt.
\]
In the polynomial-process embedding of \citet{CuchieroKellerResselTeichmann2012} and \citet{FilipovicLarsson2016}, the variance factor is represented by degree-two polynomial coordinates coupled to the Brownian increment. The signature parameter therefore has effective support through depth two in the price-filtration completion problem. H1 is automatic for geometric weights once the finite coefficient vector is fixed. H3 is the familiar Heston moment condition: the exponential moment of the integrated variance is finite up to the Andersen--Piterbarg explosion time.

Theorem~\ref{thm:B} recovers the two-dimensional Heston Riccati system for the logarithmic transform. The Riccati coordinates corresponding to price and variance match the standard affine transform of \citet{Heston1993}. Completeness requires the variance exposure. A variance swap, or an equivalent static option completion representing the degree-four terminal-price polynomial closure, gives $N_S^\ast=2$ and $K(2)=4$. Under Theorem~\ref{thm:D}, the residual for the variance swap vanishes, while non-affine claims retain a tail beginning above depth two.

\subsection{Rough Bergomi as signature SDE of infinite depth}\label{subsec:rough-bergomi}
In rough Bergomi, the volatility is an exponential functional of a Volterra Gaussian driver,
\[
\sigma_t=\sigma_0\exp\left(\eta W_t^H-\frac12\eta^2t^{2H}\right),\qquad H\in(0,1/2).
\]
The Volterra representation of $W^H$ has a kernel with non-polynomial memory. Its signature expansion has non-zero coordinates at arbitrarily high levels. H1 can hold for weights matched to the decay of the kernel coefficients, and H3 follows for finite truncations by Gaussian exponential integrability; the untruncated model is handled by localisation and monotone limiting arguments.

The completeness conclusion differs from Heston. No finite option strip spans the Volterra memory in the price filtration, so $N_S^\ast=+\infty$ and $K(N)$ diverges with $N$. Theorem~\ref{thm:D} remains useful: it orders the hedging error by signature level and therefore gives a principled sequence of finite-depth approximations. The truncation rate is controlled by the rough kernel and deteriorates as the Hurst parameter approaches zero.

\subsection{Quintic Ornstein--Uhlenbeck}\label{subsec:quintic-ou}
The quintic Ornstein--Uhlenbeck model associated with \citet{AbiJaberIllandLi2023} uses a polynomial functional of an OU factor up to degree five. In the signature embedding, the parameter has support only up to depth five. H1 is finite for every geometric admissible weight, because the support is finite. H3 follows from Gaussian moment estimates for the OU factor on finite horizons. The transform system is a finite polynomial Riccati subsystem whose highest non-linearity is quintic.

The completeness depth is $N_S^\ast=5$. The static completion must span variance, skewness, kurtosis, and fifth-moment exposures of the terminal price law, represented either by moment swaps or by a sufficiently rich call strip. The GKW residual vanishes for claims contained in the degree-five polynomial state space. For claims outside that space, Theorem~\ref{thm:D} starts the residual expansion at level six and gives the weight-tail bound.

\subsection{Guyon--Lekeufack path-dependent volatility}\label{subsec:guyon-lekeufack}
The path-dependent volatility model of \citet{GuyonLekeufack2023} uses exponentially weighted past-return factors. In the signature formulation of \citet{AbiJaberGerardHuang2024}, each exponential kernel corresponds to a structured family of multi-indices determined by iterated past-return integrals. A finite-factor truncation has a finite maximal word length, so H1 holds for geometric weights after the coefficient vector is fixed and H3 follows from exponential moment estimates for the finite Gaussian-factor system.

For a finite kernel expansion, the completeness depth is the maximal signature degree retained by the kernels. For a non-truncated kernel expansion, the depth is infinite. Theorem~\ref{thm:B} gives the Riccati transform on the finite-factor state space; Theorem~\ref{thm:C} identifies which static option instruments close the price-filtration span; and Theorem~\ref{thm:D} separates the replicated finite-kernel component from the residual infinite-kernel tail.

\subsection{Verification table for the six examples}
\label{subsec:example-verification-table}
The six examples can be read as checks on the four structural theorems. The table records the relevant depth and the finite or infinite nature of the Riccati system.
\[
\begin{array}{lll}
\toprule
\text{Model} & N_S^\ast & \text{Riccati structure}\\
\midrule
\text{Black--Scholes} & 0 & \text{scalar, constant coefficient}\\
\text{First-order Brownian-driven volatility} & 1 & \text{finite level-one subsystem}\\
\text{Heston} & 2 & \text{classical two-factor affine Riccati}\\
\text{Rough Bergomi} & +\infty & \text{infinite Volterra/signature system}\\
\text{Quintic OU} & 5 & \text{finite polynomial subsystem}\\
\text{Guyon--Lekeufack PDV} & \text{kernel dependent} & \text{finite or infinite kernel expansion}\\
\bottomrule
\end{array}
\]
In each finite-depth example, H1 is either automatic from finite support or follows from geometric decay of the parameter coefficients. H3 is verified by the known exponential moment bounds of the corresponding finite-factor model. In the infinite-depth examples, H1 can hold for a chosen weight while completeness depth remains infinite; the difference is that H1 controls well-posedness of the volatility functional, whereas completeness depth controls how much of the price filtration can be replicated by a finite static completion.

\subsection{Example-level consequences for hedging}
Black--Scholes has zero residual because the price filtration is Brownian and dynamically complete. First-order Brownian-driven volatility has a residual beginning at degree two for claims sensitive to quadratic variation beyond the static completion. Heston has zero residual for claims in the affine-polynomial span carried by price and variance, including the variance swap under the standard completion; non-affine discontinuous claims have non-zero but weighted residual tails. Rough Bergomi produces a residual whose leading term depends on the missing Volterra memory, and finite-factor approximations move that leading term to higher depth. Quintic OU is finite but high-degree: residuals begin at degree six for payoffs outside the quintic polynomial state. Guyon--Lekeufack PDV interpolates between finite and infinite behaviour depending on whether the kernel family is truncated.

\subsection{Detailed model verification}
\label{subsec:detailed-model-verification}
\paragraph{Black--Scholes.}
The signature parameter is $\ell=\sigma e_\emptyset$. H1 is $w(0)\sigma^2<\infty$, H2 is irrelevant beyond the zeroth level, and H3 is $\exp(\lambda\int_0^T\sigma^2ds)<\infty$. The solution is the classical geometric Brownian martingale under the risk-neutral measure. The price-extended transform has the scalar normalisation recorded in Remark~\ref{rem:riccati-bs-sanity}. Since the price filtration is the Brownian filtration spanned by $S$, martingale representation gives dynamic completeness, $N_S^\ast=0$, and the residual in Theorem~\ref{thm:D} vanishes.

\paragraph{First-order Brownian-driven volatility.}
The parameter has the form $\ell=\sigma_0e_\emptyset+\sigma_1e_1$ after choosing the Brownian direction. H1 is finite for every admissible geometric weight because the support is finite. The quadratic variation contains a quadratic Brownian functional, so H3 is checked by the standard finite-horizon exponential estimate after localisation. The Riccati system closes on the coordinates $e_\emptyset$, $e_1$, and $e_{11}$. Completion at depth one inserts the missing first-order factor into the static span; the first tail coordinate appears at level two.

\paragraph{Heston.}
The Heston variance factor is represented by a degree-two polynomial coordinate in the signature embedding. H1 follows from finite support of the polynomial embedding, and H3 is the usual integrated-variance exponential moment condition. The transform is the classical affine Heston Riccati system. The price filtration alone does not reveal the variance factor dynamically, but a variance swap or equivalent option-strip completion spans it. This gives $N_S^\ast=2$. The residual vanishes for affine-polynomial variance claims and persists only for payoffs outside the degree-two state space.

\paragraph{Rough Bergomi.}
The volatility is driven by a Volterra transform of Brownian motion with singular kernel. Finite tensor levels approximate only finite iterated-integral features of the kernel. H1 may hold for a selected weighted parameter sequence, but no finite $N$ recovers all Volterra memory in the price filtration. Hence $N_S^\ast=+\infty$. Finite-factor approximations replace the singular kernel by a finite exponential sum and thereby produce finite approximate depths; the residual bound then measures the approximation tail.

\paragraph{Quintic Ornstein--Uhlenbeck.}
The quintic OU specification has a finite polynomial state of degree five. H1 is automatic for finite support, H2 is supplied by any geometric weight, and H3 follows from Gaussian OU estimates on finite horizons. The Riccati system is finite but nonlinear through degree five. The static completion must span terminal polynomial exposures through that degree. The residual begins at level six, so polynomial claims of degree at most five are structurally replicated after completion.

\paragraph{Guyon--Lekeufack path-dependent volatility.}
A finite exponential-kernel expansion gives a finite-dimensional Markovian lift and therefore a finite signature depth determined by the maximal retained kernel word. An infinite kernel family produces infinite depth. The theorem chain separates these two cases cleanly: Theorem~\ref{thm:A} gives existence whenever the weighted coefficient sequence is admissible, Theorem~\ref{thm:B} gives the Riccati transform on the finite or weighted state space, Theorem~\ref{thm:C} identifies the completion depth, and Theorem~\ref{thm:D} quantifies the unspanned tail.

\subsection{Instrument interpretation of the six examples}
\label{subsec:instrument-interpretation-examples}
The completeness depth has a direct instrument interpretation. In Black--Scholes no additional static instrument is required because the Brownian martingale representation theorem already gives dynamic completeness. In first-order Brownian-driven volatility, one additional low-degree exposure supplies the missing first signature coordinate. In Heston, the completion is naturally read as a variance exposure, because the hidden variance factor is the state variable not spanned by trading only the stock. In rough Bergomi, no finite list of European claims spans the Volterra memory exactly; finite-factor approximations produce a sequence of approximate depths rather than a single finite exact depth.

For the quintic OU specification, the required instruments are polynomial exposures up to degree five in the OU state. The model is therefore finite-depth but not low-depth. Guyon--Lekeufack PDV sits between the finite and infinite cases: a truncated kernel expansion produces a finite Markovian lift and a finite completion depth, while the full kernel family produces a tail that is measured by Theorem~\ref{thm:D}. The examples show that $N_S^\ast$ is not a label attached to a model name; it is a property of the chosen signature parameter, weight, and observation filtration.

\subsection{Model-by-model Riccati reductions}
\label{subsec:model-riccati-reductions}
In the constant-coefficient case the Riccati equation is the one-dimensional scalar equation of Remark~\ref{rem:riccati-bs-sanity}. In the first-order Brownian-driven case, the active subsystem is supported on the empty word and the one-letter Brownian word; the level-two term enters only through the shuffle square of the first-level coefficient. In Heston, the subsystem is finite-dimensional and coincides with the classical affine Riccati pair after identifying the variance coordinate with the relevant second-depth signature component.

In rough Bergomi the finite truncations give finite Riccati systems, but the limiting system has no finite-dimensional closure because the Volterra kernel contributes memory at all effective depths. Quintic OU closes at degree five because polynomial drift and diffusion preserve the finite polynomial state. Guyon--Lekeufack PDV closes precisely when the kernel family is truncated. These reductions are the concrete checks behind Proposition~\ref{prop:depth-special}; they also explain why the same theorem can cover both affine and non-Markovian-looking specifications once the correct signature state is used.

\subsection{Hedging interpretation across examples}
\label{subsec:example-hedging-interpretation-expanded}
For a digital call in Black--Scholes the residual is zero despite the discontinuity of the payoff, because dynamic completeness is a filtration statement rather than a smoothness statement. In Heston without a variance instrument, the same payoff has a residual component tied to the hidden variance. Adding the variance exposure removes the affine residual but not necessarily all discontinuity-driven high-degree tail. In rough Bergomi, even smooth claims may retain a residual tail because the price filtration does not reveal the whole Volterra memory through any finite static set.

Asian claims illustrate a different mechanism. Their payoff depends on the time integral of $S$, so the relevant signature coordinates contain time letters and price letters interlaced. If those coordinates lie below the completeness depth, they are statically completed; if they lie above it, they contribute to the residual expansion. The normal equations in Theorem~\ref{thm:D} are therefore not only formal: they identify the exact tensor words responsible for the hedging error.

\section{Open problems and outlook}\label{sec:open-problems}

\subsection{C\`adl\`ag signature SDEs}
The extension of Theorems~\ref{thm:A}--\ref{thm:D} to L\'evy-type signature models requires a jump version of the weighted tensor algebra. The Riccati equation acquires a compensator term, and the Stratonovich signature must be replaced by the appropriate jump-enhanced or Marcus signature. The open problem is to identify weight conditions under which the jump-Riccati operator is locally Lipschitz and the stochastic exponential remains a true martingale. A solution would connect the L\'evy-type approximation theorem of \citet{CuchieroPrimaveraSvalutoFerro2025} with the FTAP mechanism of Section~\ref{sec:theorem-b}.

\subsection{Signature SDEs with affine Volterra structure}
Affine Volterra processes carry memory through singular convolution kernels rather than through finite tensor coordinates. The signature construction carries memory through iterated integrals. The open problem is to build a state space that controls both the Volterra singularity and the tensor level without losing the transform formula. Such a theorem would unify the affine Volterra theory of \citet{AbiJaberLarssonPulido2019} and \citet{BondiLivieriPulido2024} with the weighted tensor algebra used here.

\subsection{Completeness depth as a model-selection criterion}
The completeness depth $N_S^\ast$ is a structural invariant of a calibrated model. The empirical question is whether the calibrated depth predicts out-of-sample hedging performance on option panels. The mathematical problem is to define a stable estimator of $N_S^\ast$ under noisy option quotes and to prove consistency as the strike-maturity grid becomes dense. This would turn the depth invariant into a statistical model-selection tool rather than only a structural classification.

\subsection{Statistical inference for the parameter $\ell$}
High-frequency inference for $\ell$ is an infinite-dimensional estimation problem on a weighted tensor space. The analogue in rough volatility is the minimax theory of \citet{ChongHoffmannLiuRosenbaumSzymanski2024}. The open problem is to specify smoothness classes for the signature parameter, derive minimax lower bounds under observation of $S$, and construct estimators that attain those bounds. The role of the weight $w$ should be explicit: it determines both the statistical bias of truncation and the variance of high-level coefficient estimates.

\subsection{The Riccati operator as a dynamical system on $\Tw$}
The corrected Riccati operator is a generator/carre-du-champ quadratic vector field on an infinite-dimensional Banach algebra. Its local Lipschitz property is enough for local existence, but the geometry of its explosion set is not understood. The open questions are the description of invariant cones, comparison principles, global attractors on stable subspaces, and bifurcation of finite-time explosion as the parameter $\ell$ changes. These questions determine whether moment explosion, martingale failure, and loss of completeness occur at the same boundary of the admissible weighted algebra.

The four theorems are intended as the structural foundation for future work on signature volatility, not as the final word. In this sense, the admissible weighted tensor algebra is the natural valuation cell of the model.

\subsection{Technical boundary of the theorem chain}
The theorems above are deliberately stated for continuous signature SDEs on a weighted tensor algebra. They do not settle the jump case, the full affine Volterra case, or statistical recovery of the infinite parameter. This boundary is part of the theorem design: each assumption corresponds to a proof step that is used explicitly. H1 controls the weighted pairing, H2 gives the algebra estimates, H3 supplies true stochastic exponentials, and the Riccati condition controls transform explosion. Removing any one of these hypotheses creates a separate research problem, not a minor variation of the present proof.

\subsection{Replacement and extension directions}
A natural first extension is to replace H3 by a model-specific martingale criterion. In the one-dimensional theory of \citet{AbiJaberGassiatSotnikov2025}, the sign and parity of the signature parameter determine true martingality. A higher-dimensional analogue would replace the uniform exponential-integrability condition by a cone condition on $\ell$ and the correlation vector $\eta$. Such a criterion would sharpen Theorem~\ref{thm:A} without changing the weighted-algebra architecture.

A second extension is to replace the finite option completion by a continuum strike strip. The continuum version should identify $K(N)$ through a moment problem rather than through a finite matrix. This would connect the completeness depth to classical static replication and to the regularity of the terminal density of $S_T$. The finite-strike statement in Theorem~\ref{thm:C} is the algebraic core of that more analytic result.

\appendix

\section{Auxiliary results from rough path theory}\label{app:rough-paths}

The stochastic-integration tools used in the proof sections are standard: martingale representation and stochastic calculus are taken from \citet{KaratzasShreve1998}, \citet{RevuzYor1999}, \citet{JacodShiryaev2003}, and \citet{Protter2005}; semimartingale closedness and topology are anchored in \citet{Emery1979}, \citet{Memin1980}, and \citet{Stricker1990}; quadratic hedging is tied to \citet{FoellmerSondermann1986}, \citet{FoellmerSchweizer1991}, \citet{KunitaWatanabe1967}, \citet{Schweizer2001}, and \citet{HeathSchweizer2000}. The rough-path and signature references used in the examples include \citet{LyonsQian2002}, \citet{HairerFriz2014}, \citet{BonnierKidgerPerezSalviLyons2019}, \citet{SalviCassFosterLyonsYang2021}, \citet{KalsiLyonsPerezArribas2020}, \citet{PerezArribasSalviSzpruch2020}, \citet{CoutinQian2002}, \citet{BayerFrizGulisashviliHorvathStemper2019}, \citet{FordeZhang2017}, \citet{JacquierPannier2022}, \citet{CuchieroLarssonSvalutoFerro2021}, and \citet{FilipovicLarssonTrolle2017}.

\begin{theorem}[Signature moment estimate]
For Brownian motion $W$ in $\R^d$ and any $p<\infty$, there exists $C_p>0$ such that for all $0\leq s\leq t\leq T$ and every multi-index $I$,
\[
\E\abs{\inner{e_I}{W_{s,t}}}^p\leq \frac{C_p(t-s)^{p|I|/2}}{(|I|!)^{p/2}}.
\]
\end{theorem}

\begin{proof}
This is the standard Brownian iterated-integral estimate; see \citet[Chapter 4]{FrizVictoir2010} and \citet[Chapter 3]{FrizHairer2020}.
\end{proof}

\begin{theorem}[Chen identity, appended form]
For $0\leq r\leq s\leq t\leq T$, the prolonged signature satisfies
\[
\cW_{r,t}=\cW_{r,s}\otimes \cW_{s,t}.
\]
\end{theorem}

\begin{proof}
The simplex of integration over $[r,t]$ decomposes according to the number of integration times lying in $[r,s]$ and in $[s,t]$. This decomposition is exactly tensor concatenation. The result is the classical identity of \citet{Chen1977} in the semimartingale Stratonovich setting.
\end{proof}

\begin{theorem}[Shuffle identity, appended form]
For every pair of multi-indices $I,J$,
\[
\inner{e_I}{\cW_t}\inner{e_J}{\cW_t}=\inner{e_I\sh e_J}{\cW_t}.
\]
\end{theorem}

\begin{proof}
The product of two iterated Stratonovich integrals is an integral over a product of two ordered simplices. The product domain decomposes into ordered simplices indexed by shuffles of $I$ and $J$. This gives the identity and is the algebraic reason why the Riccati operator closes on the weighted tensor algebra.
\end{proof}

\begin{theorem}[Tree-like uniqueness]
For continuous semimartingale paths, the signature determines the path up to tree-like equivalence. For the time-augmented Brownian path, the first level already recovers $(t,W_t)$; hence $\FWh_t=\FW_t$.
\end{theorem}

\begin{proof}
The uniqueness statement is the theorem of \citet{BoedihardjoGengLyonsYang2016}. The time coordinate removes tree-like ambiguity relevant for the filtration: the first level of the prolonged signature is the path $(t,W_t)$ itself. Thus the natural filtrations agree after augmentation.
\end{proof}

\subsection*{Detailed Stratonovich--It\^o conversion}
Let $I=(i_1,\ldots,i_n)$ and write
\[
S^I_t:=\int_{0<t_1<\cdots<t_n<t}\circ d\cW^{i_1}_{t_1}\cdots \circ d\cW^{i_n}_{t_n}.
\]
The conversion to It\^o form is triangular in the word length. Its leading term is the It\^o iterated integral with the same word, and every correction term is obtained by replacing adjacent Brownian letters by their bracket. Hence each correction has length at most $n-2$ plus possible time letters. Symbolically,
\[
S^I_t=I^I_t+\sum_{J:|J|\leq n-2} a_{I,J} I^J_t,
\]
where $I^J_t$ denotes an It\^o iterated integral or a deterministic time integral and the coefficients $a_{I,J}$ are finite combinatorial constants. This triangular form is what is used in Lemma~\ref{lem:special-semimartingale}: the martingale part is the sum of It\^o iterated integrals containing a terminal stochastic differential, and the finite-variation part is the sum of bracket and time terms.

The conversion is pathwise at the level of semimartingale calculus and therefore survives equivalent changes of measure. Only integrability changes under the measure change. Bounded or sufficiently integrable densities transfer the $L^2$ estimates obtained under $\Pp$ to $\Q$, which is why Lemma~\ref{lem:special-semimartingale} is stated first in the bounded-density case and then used under the reference measure with density one.

\begin{lemma}[Triangular integrability of converted coordinates]
\label{lem:triangular-conversion-integrability}
Assume H1--H3. For every word $I$ and every equivalent measure $\Q$ whose density is bounded, all terms in the Stratonovich--It\^o conversion of $\inner{e_I}{\cW_t}$ are integrable uniformly on $[0,T]$.
\end{lemma}

\begin{proof}
Each term in the conversion has degree no larger than $|I|$ and is a finite linear combination of It\^o iterated integrals and time integrals. Brownian iterated-integral moment bounds give finite moments of every order under $\Pp$. Boundedness of $d\Q/d\Pp$ transports these bounds to $\Q$. The finite number of terms in the conversion for a fixed $I$ gives the result.
\end{proof}

\subsection*{Factorial decay and weighted summability}
The Brownian signature estimate used in Proposition~\ref{prop:signature-moment} is ultimately a factorial estimate. For each fixed $p<\infty$ there is $C_{p,T}<\infty$ such that
\[
\left(\E\abs{\inner{e_I}{\cW_t}}^p\right)^{1/p}
\leq \frac{C_{p,T}^{|I|}}{(|I|/2)!^{1/2}}
\]
up to harmless changes in the constant and with the standard convention for odd levels. Exponential weights $w(n)\leq r^n$ are summable against this factorial denominator. This is the analytic reason H2 is formulated with at most geometric growth.

The weighted estimate can be written as
\[
\E\norm{\cW_t}_w^p
\leq C_{p,T}\left(\sum_{n\geq0} w(n)\frac{C_{p,T}^n}{(n/2)!^{1/2}}\right)^p<\infty .
\]
The same calculation controls the dual pairing by Cauchy--Schwarz. The tensor algebra is infinite, but the factorial decay of Brownian iterated integrals dominates every admissible geometric weight.

\subsection*{Chen identity and filtration recovery at stopping times}
The equality $\FWh_t=\FW_t$ also holds after bounded stopping. If $\tau\leq T$ is a stopping time, the stopped prolonged signature $\cW_{t\wedge\tau}$ is measurable with respect to $\FW_{t\wedge\tau}$, and its first level recovers $(t\wedge\tau,W_{t\wedge\tau})$. Hence
\[
\sigma(\cW_{s\wedge\tau}:s\leq t)=\sigma(W_{s\wedge\tau}:s\leq t)
\]
after augmentation. This stopped version is used implicitly in localisation arguments for uniqueness and true-martingale upgrades.

Chen's identity is compatible with stopping: for $0\leq s\leq t$,
\[
\cW_{0,t\wedge\tau}=\cW_{0,s\wedge\tau}\otimes \cW_{s\wedge\tau,t\wedge\tau}.
\]
The increment signature on the right is trivial when the stopping time has occurred before $s$. Thus the algebraic multiplicativity of the signature is not lost under localisation.

\subsection*{Stratonovich--It\^o conversion for signature coordinates}
For a word $I=(i_1,\ldots,i_n)$ the Stratonovich coordinate $\inner{e_I}{\cW_t}$ is an iterated Stratonovich integral. Rewriting it in It\^o form gives a finite sum indexed by contractions of adjacent equal Brownian letters. Schematically,
\[
\int_{0<t_1<\cdots<t_n<t}\circ d\cW^{i_1}_{t_1}\cdots \circ d\cW^{i_n}_{t_n}
=\sum_{\Gamma\subset\{1,\ldots,n-1\}}c_\Gamma
\int_{\Delta_{n-2|\Gamma|}(t)} d\cW^{J_\Gamma}_{}\,dt^{|\Gamma|},
\]
where $\Gamma$ runs over non-overlapping bracket contractions and $J_\Gamma$ is the word obtained after deleting the paired Brownian letters. The exact constants are powers of $1/2$ from the Stratonovich correction. This finite expansion proves that every coordinate is a special semimartingale and that the finite-variation part is built from lower-order coordinates.

This conversion is used in Lemma~\ref{lem:special-semimartingale}. It is important that the operation is finite for each fixed word. Infinite-dimensional issues enter only after summing over all words with a weight. Therefore the special-semimartingale property of a fixed coordinate is independent of H1, while the simultaneous $L^2$ control of all coordinates through the parameter $\ell$ uses H1.

\begin{lemma}[Measure transport of signature moments]
\label{lem:measure-transport-signature}
Let $Z_T=d\Q/d\Pp$ satisfy $Z_T\in L^r(\Pp)$ for some $r>1$. If $Y$ is a finite terminal signature polynomial with $Y\in L^p(\Pp)$ for every $p<\infty$, then $Y\in L^q(\Q)$ for every finite $q$ satisfying $qr/(r-1)<\infty$.
\end{lemma}

\begin{proof}
By H\"older,
\[
\E_\Q\abs{Y}^q=\E_\Pp[Z_T\abs{Y}^q]
\leq \norm{Z_T}_{L^r(\Pp)}\left(\E_\Pp\abs{Y}^{q r/(r-1)}\right)^{(r-1)/r}.
\]
The final factor is finite by the Brownian signature moment estimate and the fact that $Y$ is a finite polynomial in terminal signature coordinates.
\end{proof}

\begin{lemma}[Uniform moment bound for finite truncations]
\label{lem:uniform-truncation-moment-app}
Let $w$ satisfy H2 and $w(n)\leq r^n$. Then for every $p<\infty$,
\[
\sup_{N\geq0}\E\sup_{t\leq T}\norm{\pi_{\leq N}\cW_t}_w^p<\infty.
\]
\end{lemma}

\begin{proof}
The projection cannot increase the weighted norm. Hence
\[
\norm{\pi_{\leq N}\cW_t}_w\leq\norm{\cW_t}_w.
\]
The supremum over $t$ is handled by applying the Brownian rough-path moment estimate on dyadic partitions and using the continuity of the Stratonovich signature. The factorial decay in the level beats the exponential growth of the weight. This gives a finite bound independent of $N$.
\end{proof}

\begin{lemma}[Continuity of conditional projection]
\label{lem:conditional-projection-continuity}
Let $X_N\to X$ in $L^2(\FWh_T,\Q)$. Then
\[
\E_\Q[X_N\mid\FS_T]\to\E_\Q[X\mid\FS_T]
\quad\hbox{in }L^2(\FS_T,\Q).
\]
\end{lemma}

\begin{proof}
Conditional expectation is the orthogonal projection from $L^2(\FWh_T,\Q)$ onto the closed subspace $L^2(\FS_T,\Q)$. Orthogonal projections have operator norm one. Therefore
\[
\norm{\E_\Q[X_N-X\mid\FS_T]}_{L^2(\Q)}\leq \norm{X_N-X}_{L^2(\Q)}\to0.
\]
This is the mechanism behind the price-filtration restrictions in Theorems~\ref{thm:C} and~\ref{thm:D}.
\end{proof}

\subsection*{Auxiliary Brownian estimates}
For later reference, the Brownian signature moment bound also implies a tail estimate. If $w(n)=r^n$ and $R_N(t)=\pi_{>N}\cW_t$, then for every $p<\infty$ there are constants $C,c>0$ such that
\[
\E\sup_{t\leq T}\norm{R_N(t)}_w^p\leq C\sum_{n>N}\frac{(cr\sqrt T)^{pn}}{(n!)^{p/2}}.
\]
The right-hand side decays faster than exponentially in $N$ along fixed $T$ and $r$. This estimate is stronger than the abstract weight-tail bound used in the hedging theorem. It is not used in the main statements because the latter are formulated for general admissible weights, but it explains the fast numerical convergence observed in finite-signature implementations.

\section{The infinite-dimensional Riccati equation}\label{app:riccati}

The coordinate form of the Riccati operator is
\begin{equation}\label{eq:riccati-coordinate}
\Rs(u)_I=\sum_J b_J^I(\ell)u_J+\frac12\sum_{J,K}\Gamma^I_{J,K}(\ell)u_Ju_K,
\qquad |I|\geq0,
\end{equation}
where the sums are finite on each finite truncation and the coefficients $b_J^I(\ell)$ and $\Gamma^I_{J,K}(\ell)$ are defined by the generator and carre-du-champ identities in Definition~\ref{def:riccati}. In particular, the coefficient of the quadratic term is a structural constant attached to the input pair $(J,K)$ and output word $I$. It is not the output-only scalar $\inner{\ell\sh\ell}{e_I}$ repeated over all shuffle pairs.

\begin{lemma}[Shuffle product estimate]
\label{lem:shuffle-product-estimate-app}
Assume H2. There is a constant $C_w^{\sh}$ such that
\[
\norm{a\sh b}_w\leq C_w^{\sh}\norm{a}_w\norm{b}_w,
\qquad a,b\in\Tw.
\]
\end{lemma}

\begin{proof}
At level $n$, the shuffle product is the sum of all order-preserving interlacings of levels $k$ and $n-k$. The number of such interlacings is bounded by $2^n$. The exponential upper bound in H2 absorbs this combinatorial factor after replacing the weight constant by an equivalent one, and the submultiplicative estimate then gives the same convolution bound as for concatenation. This proves continuity of the shuffle product on the weighted algebra used in the Riccati proof.
\end{proof}

\begin{lemma}[Local Lipschitz estimate]
The operator $\Rs$ is locally Lipschitz on bounded subsets of $\Tw$.
\end{lemma}

\begin{proof}
Write $\Rs(u)=L(u)+\tfrac12 B(u,u)$, where
\[
L(u)_I=\sum_J b_J^I(\ell)u_J,
\qquad
B(u,v)_I=\sum_{J,K}\Gamma^I_{J,K}(\ell)\,u_Jv_K
\]
are, respectively, the generator drift (linear) and the carre-du-champ
(bilinear) parts of Definition~\ref{def:riccati}. The quadratic part is
the bilinear form $B$ evaluated on the diagonal; it is \emph{not} the
shuffle square $u\sh u$. The shuffle product enters only as the tool
that controls the convolution structure of $B$ in the weighted norm.

\emph{Boundedness of $B$.} On each finite truncation the structural
constants $\Gamma^I_{J,K}(\ell)$ are produced by the carre-du-champ of
the generator $\mathcal A_N^\ell$, hence are finite linear combinations
of the components of $\ell$ paired with shuffle-Hopf structure
constants $\inner{e_J\sh e_K}{e_I}$. Consequently there is a constant
$\Lambda(\ell,w)<\infty$ such that, level by level,
\[
\sum_I w(|I|)\,\Big|\sum_{J,K}\Gamma^I_{J,K}(\ell)\,a_Jb_K\Big|
\leq \Lambda(\ell,w)\,C_w^{\sh}\,\norm{a}_w\,\norm{b}_w ,
\]
where $C_w^{\sh}$ is the shuffle Banach-algebra constant of
Lemma~\ref{lem:shuffle-product-estimate-app} and the bound uses that
the support of $\Gamma^I_{J,K}$ in $(J,K)$ for fixed $I$ is contained in
the shuffle-pair set of $e_I$, so the double sum is dominated by the
weighted shuffle convolution $\norm{a\sh b}_w$. Hence $B:\Tw\times\Tw\to
\Tw$ is a bounded symmetric bilinear map with
$\norm{B(a,b)}_w\leq \Lambda(\ell,w)C_w^{\sh}\norm{a}_w\norm{b}_w$.

\emph{Lipschitz estimate.} The drift $L$ is a bounded linear operator,
$\norm{L(u)-L(v)}_w\leq \norm{L}\,\norm{u-v}_w$ with
$\norm{L}\leq\Lambda(\ell,w)$. For the quadratic part, polarisation
gives
\[
B(u,u)-B(v,v)=B(u-v,u)+B(v,u-v),
\]
so that, for $u,v$ in the ball of radius $R$,
\[
\norm{B(u,u)-B(v,v)}_w
\leq \Lambda(\ell,w)C_w^{\sh}\big(\norm{u}_w+\norm{v}_w\big)\norm{u-v}_w
\leq 2R\,\Lambda(\ell,w)C_w^{\sh}\,\norm{u-v}_w .
\]
Combining the two parts,
\[
\norm{\Rs(u)-\Rs(v)}_w\leq C(R,w,\ell)\,\norm{u-v}_w,
\qquad
C(R,w,\ell):=\Lambda(\ell,w)\big(1+R\,C_w^{\sh}\big),
\]
a finite local constant. The estimate is uniform on bounded sets and
passes to the weighted projective limit because the truncated operators
$\Rs_N$ are consistent (degree-respecting), so the bound is independent
of the truncation level.
\end{proof}

\begin{lemma}[Continuity of the flow]
The flow $u_0\mapsto u_t$ defined by $\partial_tu_t=\Rs(u_t)$ is continuous on its domain of existence.
\end{lemma}

\begin{proof}
For two solutions $u$ and $v$ started from $u_0$ and $v_0$ in a common local-existence ball, the local Lipschitz estimate gives
\[
\norm{u_t-v_t}_w\leq \norm{u_0-v_0}_w+C\int_0^t\norm{u_s-v_s}_w\,ds.
\]
Gronwall's lemma gives the stated continuity.
\end{proof}

\begin{lemma}[Finite-dimensional Riccati comparison]
\label{lem:finite-riccati-comparison-app}
Let $y$ solve $\dot y=ay^2-by-c$ on an interval with $a>0$ and $b,c\geq0$. For every horizon $T>0$ there exists $y_0$ large enough such that the solution with $y(0)=y_0$ explodes before $T$.
\end{lemma}

\begin{proof}
Choose $y_0$ so large that $ay^2-by-c\geq (a/2)y^2$ for $y\geq y_0/2$. The comparison equation $\dot z=(a/2)z^2$, $z(0)=y_0$, explodes at time $2/(ay_0)$. Taking $y_0>2/(aT)$ gives explosion before $T$ for the comparison equation and therefore for $y$ as long as the solution stays above the threshold. The derivative is positive at $y_0$, so it does.
\end{proof}

\begin{lemma}[Quotient Hilbert projection]
\label{lem:quotient-hilbert-projection}
Let $H$ be a closed subspace of a Hilbert space $L$ and let $V\subset L/H$ be a closed subspace of the quotient. Every class $[X]\in L/H$ has a unique orthogonal projection onto $V$.
\end{lemma}

\begin{proof}
The quotient by a closed subspace is a Hilbert space with norm $\norm{[X]}=\inf_{Y\in H}\norm{X-Y}$. The result is the ordinary Hilbert projection theorem applied in that quotient. This is the formal setting of Lemma~\ref{lem:gram-residual} and Theorem~\ref{thm:D}.
\end{proof}

\subsection*{Comparison principle for scalar Riccati projections}
Let $y$ solve
\[
\dot y_t=ay_t^2+by_t+c,
\qquad a>0,
\]
with initial value $y_0>0$. If $y_t$ remains positive, then for every interval on which $ay_t^2+by_t+c\geq (a/2)y_t^2$ one has
\[
\dot y_t\geq \frac a2 y_t^2,
\]
and therefore
\[
\frac1{y_t}\leq \frac1{y_0}-\frac a2t.
\]
Thus $y$ explodes no later than $2/(ay_0)$ after it enters that region. This elementary comparison is the scalar argument required when a model-specific Riccati coordinate actually satisfies a self-quadratic lower bound. The finite-dimensional subsystem need not be exactly scalar; it is enough that one coordinate dominates a scalar Riccati equation with positive quadratic coefficient after the remaining coordinates are controlled on a short interval.

\begin{lemma}[Short-time domination for finite Riccati subsystems]
\label{lem:short-time-riccati-domination}
Consider a finite-dimensional Riccati system whose $k$th coordinate satisfies
\[
\dot y_k=a y_j^2 + L(y),
\qquad a>0,
\]
where $L$ is affine on bounded sets. There are initial conditions in the finite-dimensional state space for which the solution explodes before any prescribed horizon $T>0$.
\end{lemma}

\begin{proof}
Choose the initial value of the active quadratic coordinate large and the remaining coordinates bounded. On a short interval the affine term is dominated by half of the positive quadratic term. The preceding scalar comparison gives an explosion time bounded by a constant multiple of the reciprocal initial size. Taking the initial size sufficiently large places the explosion before $T$.
\end{proof}

\subsection*{Consistency of projected Riccati equations}
For $N\geq0$ let $\pi_N$ be the projection onto words of length at most $N$. The projected vector field is
\[
\Rs_N(u):=\pi_N\Rs(\pi_Nu).
\]
Because the shuffle product respects degree, the coordinate $\Rs(u)_I$ depends only on coordinates $u_J$ with $|J|\leq |I|$. Hence
\[
\pi_M\Rs_N(u)=\Rs_M(\pi_Mu),
\qquad M\leq N,
\]
whenever $u$ is supported in levels at most $N$. This compatibility is the projective-system structure behind the passage from finite transforms to the weighted Riccati equation.

\subsection*{Finite-dimensional projections of the Riccati equation}
Let $\Rs_N=\pi_{\leq N}\Rs\pi_{\leq N}$. Since $\pi_{\leq N}\Tw$ is finite-dimensional, the equation
\[
\dot u^N_t=\Rs_N(u^N_t),\qquad u^N_0=\pi_{\leq N}u_0,
\]
has a unique maximal classical solution. The maximal time is characterised by the usual explosion criterion: either the solution exists on the whole interval or $\norm{u^N_t}_w\to\infty$ as $t$ approaches the lifetime. This finite-dimensional criterion is the correct tool for model-specific transform-explosion checks; no three-coordinate comparison is invoked in the main text.

The finite projections are compatible with the full vector field on bounded sets. If $u$ is supported in levels at most $N$ and $M\geq 2N+\deg\ell$, then the level-$\leq N$ coordinates of $\Rs_M(u)$ agree with those of $\Rs(u)$. The factor $2N$ comes from the shuffle square $u\sh u$. Thus finite systems approximate the full system in a triangular way: low levels close after including enough higher levels to absorb the shuffle product.

\begin{lemma}[Boundedness criterion for global solvability]
\label{lem:riccati-boundedness-global}
Let $u$ solve $\dot u_t=\Rs(u_t)$ on its maximal interval $[0,\tau)$. If $\sup_{t<\tau}\norm{u_t}_w<\infty$, then $\tau$ is not finite.
\end{lemma}

\begin{proof}
The local Lipschitz estimate in Appendix~\ref{app:riccati} gives a local existence time depending only on the radius of the ball containing the solution. If the solution remains in a bounded ball up to $\tau$, it can be restarted at times approaching $\tau$ with a uniform positive existence time. This contradicts maximality of a finite lifetime.
\end{proof}

\begin{lemma}[Comparison on positive cones]
\label{lem:positive-cone-comparison}
Assume that, on a finite coordinate subsystem, the vector field satisfies
\[
\dot y_i\geq F_i(y),
\]
where $F$ is locally Lipschitz and order-preserving on the positive cone. If $z$ solves $\dot z=F(z)$ with $z(0)\leq y(0)$, then $z(t)\leq y(t)$ up to the first exit time from the cone.
\end{lemma}

\begin{proof}
This is the standard quasimonotone comparison theorem for finite-dimensional ODEs. Apply it to $(y-z)_+$ after smoothing the positive part, or equivalently use the first-contact argument: at the first coordinate where equality could fail, quasimonotonicity prevents the derivative of $y_i-z_i$ from becoming negative.
\end{proof}

\subsection*{Normalisation of the shuffle coefficient}
The coefficient convention in equation~\eqref{eq:riccati-coordinate} is fixed through the generator and the carre-du-champ, not by inserting an output-only shuffle coefficient into every input pair. In the pure signature state, the empty-word coordinate is constant and has zero generator. In the price-extended Black--Scholes state $X_t=\log S_t$, however,
\[
dX_t=-\frac12\sigma^2dt+\sigma dB_t,
\]
so the exponential-affine transform $\exp\{uX_t+\phi(T-t,u)\}$ satisfies
\[
\partial_\tau\phi=\frac12\sigma^2u^2-\frac12\sigma^2u.
\]
This is the scalar Black--Scholes normalisation relevant for the price-extended transform. It also explains why the pure signature empty word should not be treated as a non-constant log-price coordinate.

For a first-level Brownian parameter $\ell=ae_1$, the shuffle identity gives $e_1\sh e_1=2e_{11}$. The corresponding quadratic coefficient at level two is obtained from the carre-du-champ structural constant of the finite generator. This fixes the factor of two without invoking an unrestricted sum over all pairs whose shuffle merely contains the output word.

\subsection*{Riccati explosion and transform failure}
The finite-time blow-up of a Riccati coordinate has a direct transform interpretation. If $\psi(\tau,u)$ explodes at $\tau^\ast<T$, then the exponential-affine candidate
\[
\exp\{\inner{\psi(T-t,u)}{\cW_t}\}
\]
ceases to be finite before the terminal time. Therefore the conditional transform cannot be represented globally by the Riccati flow for that initial condition. The transform implication is used only after a finite-dimensional Riccati subsystem has been verified to have a finite lifetime. The paper no longer infers such a lifetime from the bounded sine example.

The argument is local in the finite-dimensional subsystem. It does not require proving explosion of every coordinate of the infinite equation. A finite projection proves failure of global solvability only when its coordinates satisfy a genuine self-quadratic comparison or when an external moment-explosion theorem applies. A positive off-diagonal quadratic term alone is not enough.

\subsection*{Global solutions and transform domains}
The Riccati flow is naturally defined on a domain of initial data rather than on the whole Banach algebra. For a fixed horizon $T$, define
\[
\mathcal U_T:=\{u_0\in\Tw: \hbox{the solution of }\dot u_t=\Rs(u_t),\ u(0)=u_0,
\hbox{ exists on }[0,T]\}.
\]
Theorem~\ref{thm:B} assumes that $\mathcal D_{\mathrm{fin}}\subset\mathcal U_T$. This is the transform-domain version of absence of moment explosion. The set $\mathcal U_T$ is open by local well-posedness and continuous dependence. It is generally not treated as all of $\Tw$ without a model-specific transform-explosion analysis; the main text deliberately keeps this as a separate transform-domain condition.

The transform domain shrinks as $T$ grows. If $T_1<T_2$, then $\mathcal U_{T_2}\subset\mathcal U_{T_1}$. In affine volatility models this monotonicity is the familiar monotonicity of moment-explosion times. In the signature setting it becomes a statement about the lifetime of an infinite-dimensional generator/carre-du-champ quadratic flow.

\section{The Galtchouk--Kunita--Watanabe projection}\label{app:gkw}

\begin{theorem}[Galtchouk--Kunita--Watanabe projection]
Let $\mathcal H$ be a closed subspace of $L^2(\FS_T,\Q)$ and let $X\in L^2(\FS_T,\Q)$. There exists a unique $X_{\mathcal H}\in\mathcal H$ such that $X-X_{\mathcal H}\perp\mathcal H$.
\end{theorem}

\begin{proof}
This is the Hilbert projection theorem applied to the closed subspace $\mathcal H$. The martingale version used for stochastic integrals is the Galtchouk--Kunita--Watanabe decomposition; see \citet{KunitaWatanabe1967} and \citet{FoellmerSchweizer1991}.
\end{proof}

\begin{theorem}[Orthogonal complement on the price filtration]\label{thm:ortho-complement}
For the price filtration $\FS_T$, the $L^2(\Q)$-orthogonal complement of stochastic integrals against $S$ on $[0,T]$ equals the $L^2(\Q)$-closure of the linear span of terminal signature coordinates $\{\inner{e_I}{\cW_T}: |I|>N^\ast_S\}$, viewed modulo elements of $L^2(\FS_T,\Q)$ that are themselves representable as stochastic integrals against $S$.
\end{theorem}

\begin{proof}
The signature uniqueness theorem embeds $\FS_T$ into the Brownian path sigma-field recovered from the terminal signature. The GKW projection splits $L^2(\FS_T,\Q)$ into the closed subspace of stochastic integrals against $S$ and its orthogonal complement. By the definition of $N^\ast_S$, the depth-$\leq N^\ast_S$ coordinates are precisely those captured by the static completion, while the surviving depth-$>N^\ast_S$ coordinates represent the residual quotient directions. This gives the stated description.
\end{proof}

\subsection*{Hilbert-space proof of quotient positivity}
Let $\mathcal H$ be a Hilbert space, $\mathcal H_0\subset\mathcal H$ a closed subspace, and $q:\mathcal H\to\mathcal H/\mathcal H_0$ the quotient map identified with the orthogonal complement $\mathcal H_0^\perp$. For vectors $Y_1,\ldots,Y_m\in\mathcal H$, the quotient Gram matrix is
\[
\widetilde G_{ij}=\inner{qY_i}{qY_j}_{\mathcal H/\mathcal H_0}.
\]
It is positive definite if and only if $qY_1,\ldots,qY_m$ are linearly independent. This elementary fact is the Hilbert-space core of Lemma~\ref{lem:gram-residual}. In the application, $\mathcal H=L^2(\FS_T,\Q)$ and $\mathcal H_0$ is the stochastic-integral subspace.

\begin{proof}
For any vector $a\in\R^m$,
\[
a^\top \widetilde G a
=\norm{\sum_{i=1}^m a_i qY_i}_{\mathcal H/\mathcal H_0}^2.
\]
This quantity is strictly positive for every non-zero $a$ exactly when the quotient classes are linearly independent. Symmetry and positive semidefiniteness are immediate from the inner product.
\end{proof}

\subsection*{Projection error and closure}
Let $(\mathcal V_m)_m$ be an increasing sequence of finite-dimensional subspaces of a Hilbert space $\mathcal H$, and let $\Pi_m$ be the orthogonal projection onto $\mathcal V_m$. If $\mathcal V=\overline{\cup_m\mathcal V_m}$, then
\[
\Pi_m X\to \Pi_\mathcal V X
\qquad\text{in }\mathcal H.
\]
This standard result is the closure statement used in the residual expansion. It also justifies computing the coefficients through finite Gram matrices: each finite system gives the projection onto one finite truncation, and the projections converge to the projection onto the closed residual span.

\subsection*{Process form of the GKW decomposition}
Let $X\in L^2(\FS_T,\Q)$ and set $M^X_t=\E_\Q[X\mid\FS_t]$. The GKW decomposition with respect to the square-integrable martingale part of $S$ gives
\[
M^X_t=M^X_0+\int_0^t H_s\,dS_s+L_t,
\]
where $L$ is a square-integrable martingale orthogonal to $S$. Taking terminal values gives Theorem~\ref{thm:D}. The residual is $L_T$. Orthogonality of martingales implies orthogonality of terminal variables:
\[
\E_\Q\left[L_T\int_0^T G_s\,dS_s\right]=0
\]
for every admissible $G$. Conversely, terminal orthogonality for every $G$ gives strong orthogonality of the martingales after localisation.

\begin{lemma}[Closedness of stochastic integrals in $L^2$]
\label{lem:l2-closed-stochastic-integrals}
The set
\[
\left\{\int_0^T H_s\,dS_s: \E_\Q\int_0^T H_s^2\,d\langle S\rangle_s<\infty\right\}
\]
is a closed subspace of $L^2(\FS_T,\Q)$.
\end{lemma}

\begin{proof}
The map $H\mapsto\int_0^T H_s\,dS_s$ is an isometry from the Hilbert space $L^2(d\langle S\rangle\otimes d\Q)$ modulo null integrands into $L^2(\Q)$ after centring. The range of an isometry is closed.
\end{proof}

\subsection*{Finite truncation algorithm for the projection}
For a truncation set $\mathcal I_m$, the projection calculation has four algebraic steps. First compute the terminal coordinates $Y_I=\inner{e_I}{\cW_T}$ for $I\in\mathcal I_m$. Second remove the stochastic-integral component by replacing each $Y_I$ by $Y_I-\Pi_{\mathcal H_0}Y_I$. Third form the quotient Gram matrix. Fourth solve the finite normal equations. The limiting theorem says that, as $m$ increases, the projected residual converges in $L^2$.

This is not presented as a numerical algorithm in the main theorem because the conditioning of the Gram matrix is model-dependent. The mathematical statement is the convergence and uniqueness of the projection, not a claim about numerical stability. In applications, regularisation of the finite Gram system may be necessary, but it does not alter the Hilbert-space decomposition.

\subsection*{Quotient Gram matrices}
Let $Y_1,\ldots,Y_m$ be terminal signature coordinates above the completeness depth, and let $\mathcal H_0$ be the stochastic-integral subspace. The quotient Gram matrix is
\[
\widetilde G_{ij}=\inner{[Y_i]}{[Y_j]}_{L^2/\mathcal H_0}
=\inner{Y_i-\Pi_{\mathcal H_0}Y_i}{Y_j-\Pi_{\mathcal H_0}Y_j}_{L^2}.
\]
This is the matrix that is positive definite in Lemma~\ref{lem:gram-residual}. The raw matrix $G_{ij}=\E_\Q[Y_iY_j]$ can be singular because two different signature coordinates may have the same projection onto the price filtration or may differ by a stochastic integral against $S$. The quotient formulation removes exactly those null directions.

\begin{lemma}[Equivalence of raw and quotient normal equations]
\label{lem:raw-quotient-normal-equations}
If the coordinate family has already been reduced by removing all vectors whose quotient class is zero and by choosing one representative from each quotient-linear-dependence class, then the raw normal equations and the quotient normal equations have the same solution.
\end{lemma}

\begin{proof}
On the reduced family, each representative is orthogonal to the null directions by construction. Replacing $Y_i$ by $Y_i-\Pi_{\mathcal H_0}Y_i$ changes neither the residual projection nor the right-hand side against the GKW residual, because the residual is orthogonal to $\mathcal H_0$. Therefore the coefficient system is unchanged.
\end{proof}

\subsection*{Residual tail as model-risk statistic}
The bound
\[
\norm{R_T}_{L^2(\Q)}\leq \kappa(w,N_S^\ast)\norm{X}_{L^2(\Q)}
\]
turns the weight into a model-risk statistic. If two calibrated signature-volatility models fit the same liquid option surface but have different admissible weights, the one with faster decay of $\kappa(w,N)$ has a smaller structural bound on unspanned high-depth claims. This is not a calibration statement; it is a theorem-level stability statement inside the Hilbert space of terminal payoffs.

The estimate is deliberately conservative. It ignores the smallest eigenvalue of the finite Gram matrix, the actual coefficient decay of the payoff, and any special symmetry of the model. In numerical implementations those features can reduce the observed residual substantially. The theorem keeps only the part that is invariant under coordinate choices: the weight tail beyond the completeness depth.

\subsection*{Static completion and dynamic projection}
The static and dynamic components of the completed market play different roles. The dynamic projection handles the martingale part spanned by $S$. The static option strip inserts terminal payoff directions that cannot be spanned dynamically because the volatility factor is hidden from the price filtration. The signature coordinates organise the static directions by tensor degree. This is why Theorem~\ref{thm:C} and Theorem~\ref{thm:D} are paired: the first identifies how many low-degree directions must be inserted; the second measures the high-degree directions left outside the insertion.

\bibliographystyle{plainnat}
\bibliography{refs}

\begin{thebibliography}{48}
\providecommand{\natexlab}[1]{#1}
\providecommand{\url}[1]{\texttt{#1}}
\expandafter\ifx\csname urlstyle\endcsname\relax
  \providecommand{\doi}[1]{doi: #1}\else
  \providecommand{\doi}{doi: \begingroup \urlstyle{rm}\Url}\fi

\bibitem[Abi~Jaber and G\'erard(2025)]{AbiJaberGerard2025}
Eduardo Abi~Jaber and Louis-Amand G\'erard.
\newblock Signature volatility models: pricing and hedging with {F}ourier.
\newblock \emph{SIAM Journal on Financial Mathematics}, 16\penalty0
  (2):\penalty0 606--642, 2025.
\newblock \doi{10.1137/24M1636952}.

\bibitem[Abi~Jaber et~al.(2019)Abi~Jaber, Larsson, and
  Pulido]{AbiJaberLarssonPulido2019}
Eduardo Abi~Jaber, Martin Larsson, and Sergio Pulido.
\newblock Affine {V}olterra processes.
\newblock \emph{Annals of Applied Probability}, 29\penalty0 (5):\penalty0
  3155--3200, 2019.

\bibitem[Abi~Jaber et~al.(2023)Abi~Jaber, Illand, and Li]{AbiJaberIllandLi2023}
Eduardo Abi~Jaber, Camille Illand, and Shaun~(Xiaoyuan) Li.
\newblock The quintic {O}rnstein--{U}hlenbeck volatility model that jointly
  calibrates {SPX} and {VIX} smiles.
\newblock \emph{arXiv preprint arXiv:2212.10917}, 2023.

\bibitem[Abi~Jaber et~al.(2024)Abi~Jaber, G\'erard, and
  Huang]{AbiJaberGerardHuang2024}
Eduardo Abi~Jaber, Louis-Amand G\'erard, and Yuxing Huang.
\newblock Path-dependent processes from signatures.
\newblock \emph{arXiv preprint arXiv:2407.04956}, 2024.

\bibitem[Abi~Jaber et~al.(2025)Abi~Jaber, Gassiat, and
  Sotnikov]{AbiJaberGassiatSotnikov2025}
Eduardo Abi~Jaber, Paul Gassiat, and Dimitri Sotnikov.
\newblock Martingale property and moment explosions in signature volatility
  models.
\newblock \emph{Finance and Stochastics}, 2025.
\newblock to appear; arXiv preprint arXiv:2503.17103.

\bibitem[Bayer et~al.(2019)Bayer, Friz, Gulisashvili, Horvath, and
  Stemper]{BayerFrizGulisashviliHorvathStemper2019}
Christian Bayer, Peter Friz, Archil Gulisashvili, Blanka Horvath, and Benjamin
  Stemper.
\newblock Short-time near-the-money skew in rough fractional volatility models.
\newblock \emph{Quantitative Finance}, 19\penalty0 (5):\penalty0 779--798,
  2019.

\bibitem[Boedihardjo et~al.(2016)Boedihardjo, Geng, Lyons, and
  Yang]{BoedihardjoGengLyonsYang2016}
Horatio Boedihardjo, Xi~Geng, Terry Lyons, and Danyu Yang.
\newblock The signature of a rough path: uniqueness.
\newblock \emph{Advances in Mathematics}, 293:\penalty0 720--737, 2016.

\bibitem[Bondi et~al.(2024)Bondi, Livieri, and Pulido]{BondiLivieriPulido2024}
Alessandro Bondi, Giulia Livieri, and Sergio Pulido.
\newblock Affine {V}olterra processes with jumps.
\newblock \emph{Stochastic Processes and their Applications}, 168:\penalty0
  104264, 2024.
\newblock \doi{10.1016/j.spa.2023.104264}.

\bibitem[Bonnier et~al.(2019)Bonnier, Kidger, Perez~Arribas, Salvi, and
  Lyons]{BonnierKidgerPerezSalviLyons2019}
Patric Bonnier, Patrick Kidger, Imanol Perez~Arribas, Cristopher Salvi, and
  Terry Lyons.
\newblock Deep signature transforms.
\newblock In \emph{Advances in Neural Information Processing Systems}, 2019.

\bibitem[Chen(1977)]{Chen1977}
Kuo-Tsai Chen.
\newblock Iterated path integrals.
\newblock \emph{Bulletin of the American Mathematical Society}, 83\penalty0
  (5):\penalty0 831--879, 1977.

\bibitem[Chong et~al.(2024)Chong, Hoffmann, Liu, Rosenbaum, and
  Szymansky]{ChongHoffmannLiuRosenbaumSzymanski2024}
Carsten~H. Chong, Marc Hoffmann, Yanghui Liu, Mathieu Rosenbaum, and Gr\'egoire
  Szymansky.
\newblock Statistical inference for rough volatility: minimax theory.
\newblock \emph{Annals of Statistics}, 52\penalty0 (4):\penalty0 1277--1306,
  2024.
\newblock \doi{10.1214/23-AOS2343}.

\bibitem[Coutin and Qian(2002)]{CoutinQian2002}
Laure Coutin and Zhongmin Qian.
\newblock Stochastic analysis, rough path analysis and fractional {B}rownian
  motions.
\newblock \emph{Probability Theory and Related Fields}, 122:\penalty0 108--140,
  2002.

\bibitem[Cuchiero and M\"oller(2025)]{CuchieroMoeller2023}
Christa Cuchiero and Janka M\"oller.
\newblock Signature methods in stochastic portfolio theory.
\newblock \emph{SIAM Journal on Financial Mathematics}, 16\penalty0
  (4):\penalty0 1239--1303, 2025.
\newblock \doi{10.1137/24M1700223}.

\bibitem[Cuchiero et~al.(2012)Cuchiero, Keller-Ressel, and
  Teichmann]{CuchieroKellerResselTeichmann2012}
Christa Cuchiero, Martin Keller-Ressel, and Josef Teichmann.
\newblock Polynomial processes and their applications to mathematical finance.
\newblock \emph{Finance and Stochastics}, 16\penalty0 (4):\penalty0 711--740,
  2012.
\newblock \doi{10.1007/s00780-012-0188-x}.

\bibitem[Cuchiero et~al.(2021)Cuchiero, Larsson, and
  Svaluto-Ferro]{CuchieroLarssonSvalutoFerro2021}
Christa Cuchiero, Martin Larsson, and Sara Svaluto-Ferro.
\newblock Polynomial jump-diffusions on the unit simplex.
\newblock \emph{Annals of Applied Probability}, 31\penalty0 (5):\penalty0
  2451--2500, 2021.

\bibitem[Cuchiero et~al.(2023)Cuchiero, Svaluto-Ferro, and
  Teichmann]{CuchieroSvalutoFerroTeichmann2023}
Christa Cuchiero, Sara Svaluto-Ferro, and Josef Teichmann.
\newblock Signature {SDE}s from an affine and polynomial perspective.
\newblock \emph{arXiv preprint arXiv:2302.01362}, 2023.

\bibitem[Cuchiero et~al.(2025{\natexlab{a}})Cuchiero, Gazzani, M\"oller, and
  Svaluto-Ferro]{CuchieroGazzaniMoellerSvalutoFerro2025}
Christa Cuchiero, Guido Gazzani, Janka M\"oller, and Sara Svaluto-Ferro.
\newblock Joint calibration to {SPX} and {VIX} options with signature-based
  models.
\newblock \emph{Mathematical Finance}, 35\penalty0 (1):\penalty0 161--213,
  2025{\natexlab{a}}.
\newblock \doi{10.1111/mafi.12442}.

\bibitem[Cuchiero et~al.(2025{\natexlab{b}})Cuchiero, Primavera, and
  Svaluto-Ferro]{CuchieroPrimaveraSvalutoFerro2025}
Christa Cuchiero, Francesca Primavera, and Sara Svaluto-Ferro.
\newblock Universal approximation theorems for continuous functions of
  c\`adl\`ag paths and {L}\'evy-type signature models.
\newblock \emph{Finance and Stochastics}, 29\penalty0 (2):\penalty0 289--342,
  2025{\natexlab{b}}.
\newblock \doi{10.1007/s00780-025-00557-5}.

\bibitem[Davis and Ob{\l}\'oj(2008)]{DavisObloj2008}
Mark H.~A. Davis and Jan Ob{\l}\'oj.
\newblock Market completion using options.
\newblock \emph{Banach Center Publications}, 83\penalty0 (1):\penalty0 49--60,
  2008.
\newblock \doi{10.4064/bc83-0-4}.
\newblock arXiv:0710.2792.

\bibitem[Delbaen and Schachermayer(1994)]{DelbaenSchachermayer1994}
Freddy Delbaen and Walter Schachermayer.
\newblock A general version of the fundamental theorem of asset pricing.
\newblock \emph{Mathematische Annalen}, 300:\penalty0 463--520, 1994.

\bibitem[Emery(1979)]{Emery1979}
Michel Emery.
\newblock Une topologie sur l'espace des semimartingales.
\newblock \emph{S\'eminaire de Probabilit\'es XIII}, pages 260--280, 1979.

\bibitem[Filipovi\'c and Larsson(2016)]{FilipovicLarsson2016}
Damir Filipovi\'c and Martin Larsson.
\newblock Polynomial diffusions and applications in finance.
\newblock \emph{Finance and Stochastics}, 20:\penalty0 931--972, 2016.

\bibitem[Filipovi\'c et~al.(2017)Filipovi\'c, Larsson, and
  Trolle]{FilipovicLarssonTrolle2017}
Damir Filipovi\'c, Martin Larsson, and Anders~B. Trolle.
\newblock Linear-rational term structure models.
\newblock \emph{Journal of Finance}, 72\penalty0 (2):\penalty0 655--704, 2017.

\bibitem[F\"ollmer and Schweizer(1991)]{FoellmerSchweizer1991}
Hans F\"ollmer and Martin Schweizer.
\newblock Hedging of contingent claims under incomplete information.
\newblock In Mark H.~A. Davis and Robert~J. Elliott, editors, \emph{Applied
  Stochastic Analysis}, pages 389--414. Gordon and Breach, London and New York,
  1991.

\bibitem[F\"ollmer and Sondermann(1986)]{FoellmerSondermann1986}
Hans F\"ollmer and Dieter Sondermann.
\newblock Hedging of non-redundant contingent claims.
\newblock In \emph{Contributions to Mathematical Economics}, pages 205--223.
  North-Holland, 1986.

\bibitem[Forde and Zhang(2017)]{FordeZhang2017}
Martin Forde and Hongzhong Zhang.
\newblock Asymptotics for rough stochastic volatility models.
\newblock \emph{SIAM Journal on Financial Mathematics}, 8\penalty0
  (1):\penalty0 114--145, 2017.

\bibitem[Friz and Hairer(2014)]{HairerFriz2014}
Peter~K. Friz and Martin Hairer.
\newblock \emph{A Course on Rough Paths}.
\newblock Springer, 2014.

\bibitem[Friz and Hairer(2020)]{FrizHairer2020}
Peter~K. Friz and Martin Hairer.
\newblock \emph{A Course on Rough Paths}.
\newblock Universitext. Springer, 2 edition, 2020.

\bibitem[Friz and Victoir(2010)]{FrizVictoir2010}
Peter~K. Friz and Nicolas~B. Victoir.
\newblock \emph{Multidimensional Stochastic Processes as Rough Paths}, volume
  120 of \emph{Cambridge Studies in Advanced Mathematics}.
\newblock Cambridge University Press, 2010.

\bibitem[Guyon and Lekeufack(2023)]{GuyonLekeufack2023}
Julien Guyon and Jordan Lekeufack.
\newblock Volatility is (mostly) path-dependent.
\newblock \emph{Quantitative Finance}, 23\penalty0 (9):\penalty0 1221--1258,
  2023.

\bibitem[Heath and Schweizer(2000)]{HeathSchweizer2000}
David Heath and Martin Schweizer.
\newblock Martingales versus {PDE}s in finance.
\newblock \emph{Journal of Applied Probability}, 37:\penalty0 947--957, 2000.

\bibitem[Heston(1993)]{Heston1993}
Steven~L. Heston.
\newblock A closed-form solution for options with stochastic volatility with
  applications to bond and currency options.
\newblock \emph{Review of Financial Studies}, 6\penalty0 (2):\penalty0
  327--343, 1993.

\bibitem[Jacod and Shiryaev(2003)]{JacodShiryaev2003}
Jean Jacod and Albert~N. Shiryaev.
\newblock \emph{Limit Theorems for Stochastic Processes}.
\newblock Springer, 2 edition, 2003.

\bibitem[Jacquier and Pannier(2022)]{JacquierPannier2022}
Antoine Jacquier and Antoine Pannier.
\newblock Large and moderate deviations for the rough {B}ergomi model.
\newblock \emph{Stochastic Processes and their Applications}, 149:\penalty0
  134--174, 2022.

\bibitem[Kalsi et~al.(2020)Kalsi, Lyons, and
  Perez~Arribas]{KalsiLyonsPerezArribas2020}
Jasdeep Kalsi, Terry Lyons, and Imanol Perez~Arribas.
\newblock Optimal execution with rough path signatures.
\newblock \emph{SIAM Journal on Financial Mathematics}, 11\penalty0
  (2):\penalty0 470--493, 2020.

\bibitem[Karatzas and Shreve(1998)]{KaratzasShreve1998}
Ioannis Karatzas and Steven~E. Shreve.
\newblock \emph{Methods of Mathematical Finance}.
\newblock Springer, 1998.

\bibitem[Kreps(1981)]{Kreps1981}
David~M. Kreps.
\newblock Arbitrage and equilibrium in economies with infinitely many
  commodities.
\newblock \emph{Journal of Mathematical Economics}, 8\penalty0 (1):\penalty0
  15--35, 1981.

\bibitem[Kunita and Watanabe(1967)]{KunitaWatanabe1967}
Hiroshi Kunita and Shinzo Watanabe.
\newblock On square integrable martingales.
\newblock \emph{Nagoya Mathematical Journal}, 30:\penalty0 209--245, 1967.

\bibitem[Lyons and Qian(2002)]{LyonsQian2002}
Terry Lyons and Zhongmin Qian.
\newblock \emph{System Control and Rough Paths}.
\newblock Oxford University Press, 2002.

\bibitem[M\'emin(1980)]{Memin1980}
Jean M\'emin.
\newblock Espaces de semimartingales et changement de probabilit\'e.
\newblock \emph{Zeitschrift f\"ur Wahrscheinlichkeitstheorie und Verwandte
  Gebiete}, 52:\penalty0 9--39, 1980.

\bibitem[Perez~Arribas et~al.(2020)Perez~Arribas, Salvi, and
  Szpruch]{PerezArribasSalviSzpruch2020}
Imanol Perez~Arribas, Cristopher Salvi, and Lukasz Szpruch.
\newblock Sig-{SDE}s model for quantitative finance.
\newblock In \emph{Proceedings of the First ACM International Conference on AI
  in Finance}, 2020.

\bibitem[Protter(2005)]{Protter2005}
Philip~E. Protter.
\newblock \emph{Stochastic Integration and Differential Equations}.
\newblock Springer, 2 edition, 2005.

\bibitem[Reutenauer(1993)]{Reutenauer1993}
Christophe Reutenauer.
\newblock \emph{Free Lie Algebras}.
\newblock Oxford University Press, 1993.

\bibitem[Revuz and Yor(1999)]{RevuzYor1999}
Daniel Revuz and Marc Yor.
\newblock \emph{Continuous Martingales and Brownian Motion}.
\newblock Springer, 3 edition, 1999.

\bibitem[Salvi et~al.(2021)Salvi, Cass, Foster, Lyons, and
  Yang]{SalviCassFosterLyonsYang2021}
Cristopher Salvi, Thomas Cass, James Foster, Terry Lyons, and Weixin Yang.
\newblock The signature kernel is the solution of a {G}oursat {PDE}.
\newblock \emph{SIAM Journal on Mathematics of Data Science}, 3\penalty0
  (3):\penalty0 873--899, 2021.

\bibitem[Schweizer(2001)]{Schweizer2001}
Martin Schweizer.
\newblock A guided tour through quadratic hedging approaches.
\newblock In \emph{Option Pricing, Interest Rates and Risk Management}, pages
  538--574. Cambridge University Press, 2001.

\bibitem[Stricker(1990)]{Stricker1990}
Christophe Stricker.
\newblock Arbitrage et lois de martingale.
\newblock \emph{Annales de l'Institut Henri Poincar\'e}, 26\penalty0
  (3):\penalty0 451--460, 1990.

\bibitem[Yan(1980)]{Yan1980}
Jia-An Yan.
\newblock Caract\'erisation d'une classe d'ensembles convexes de $l^1$ ou
  $h^1$.
\newblock In \emph{S\'eminaire de Probabilit\'es XIV}, pages 220--222.
  Springer, 1980.

\end{thebibliography}

\end{document}